\PassOptionsToPackage{unicode}{hyperref}
\PassOptionsToPackage{hyphens}{url}
\PassOptionsToPackage{dvipsnames,svgnames,x11names}{xcolor}
\documentclass[
  12pt]{article}

\usepackage{amsmath,amssymb,amsthm}
\usepackage{subcaption}
\usepackage{iftex}
\ifPDFTeX
  \usepackage[T1]{fontenc}
  \usepackage[utf8]{inputenc}
  \usepackage{textcomp} 
\else 
  \usepackage{unicode-math}
  \defaultfontfeatures{Scale=MatchLowercase}
  \defaultfontfeatures[\rmfamily]{Ligatures=TeX,Scale=1}
\fi
\usepackage{lmodern}
\ifPDFTeX\else  
\fi
\IfFileExists{upquote.sty}{\usepackage{upquote}}{}
\IfFileExists{microtype.sty}{
  \usepackage[]{microtype}
  \UseMicrotypeSet[protrusion]{basicmath} 
}{}
\makeatletter
\@ifundefined{KOMAClassName}{
  \IfFileExists{parskip.sty}{%
    \usepackage{parskip}
  }{
    \setlength{\parindent}{0pt}
    \setlength{\parskip}{6pt plus 2pt minus 1pt}}
}{
  \KOMAoptions{parskip=half}}
\makeatother
\usepackage{xcolor}
\makeatletter
\ifx\paragraph\undefined\else
  \let\oldparagraph\paragraph
  \renewcommand{\paragraph}{
    \@ifstar
      \xxxParagraphStar
      \xxxParagraphNoStar
  }
  \newcommand{\xxxParagraphStar}[1]{\oldparagraph*{#1}\mbox{}}
  \newcommand{\xxxParagraphNoStar}[1]{\oldparagraph{#1}\mbox{}}
\fi
\ifx\subparagraph\undefined\else
  \let\oldsubparagraph\subparagraph
  \renewcommand{\subparagraph}{
    \@ifstar
      \xxxSubParagraphStar
      \xxxSubParagraphNoStar
  }
  \newcommand{\xxxSubParagraphStar}[1]{\oldsubparagraph*{#1}\mbox{}}
  \newcommand{\xxxSubParagraphNoStar}[1]{\oldsubparagraph{#1}\mbox{}}
\fi
\makeatother

\usepackage{longtable,booktabs,array}
\usepackage{calc} 
\usepackage{etoolbox}
\makeatletter
\patchcmd\longtable{\par}{\if@noskipsec\mbox{}\fi\par}{}{}
\makeatother
\IfFileExists{footnotehyper.sty}{\usepackage{footnotehyper}}{\usepackage{footnote}}
\makesavenoteenv{longtable}
\usepackage{graphicx}
\makeatletter
\def\maxwidth{\ifdim\Gin@nat@width>\linewidth\linewidth\else\Gin@nat@width\fi}
\def\maxheight{\ifdim\Gin@nat@height>\textheight\textheight\else\Gin@nat@height\fi}
\makeatother
\setkeys{Gin}{width=\maxwidth,height=\maxheight,keepaspectratio}
\makeatletter
\def\fps@figure{htbp}
\makeatother

\makeatletter
\@ifpackageloaded{caption}{}{\usepackage{caption}}
\AtBeginDocument{%
\ifdefined\contentsname
  \renewcommand*\contentsname{Table of contents}
\else
  \newcommand\contentsname{Table of contents}
\fi
\ifdefined\listfigurename
  \renewcommand*\listfigurename{List of Figures}
\else
  \newcommand\listfigurename{List of Figures}
\fi
\ifdefined\listtablename
  \renewcommand*\listtablename{List of Tables}
\else
  \newcommand\listtablename{List of Tables}
\fi
\ifdefined\figurename
  \renewcommand*\figurename{Figure}
\else
  \newcommand\figurename{Figure}
\fi
\ifdefined\tablename
  \renewcommand*\tablename{Table}
\else
  \newcommand\tablename{Table}
\fi
}
\@ifpackageloaded{float}{}{\usepackage{float}}
\floatstyle{ruled}
\@ifundefined{c@chapter}{\newfloat{codelisting}{h}{lop}}{\newfloat{codelisting}{h}{lop}[chapter]}
\floatname{codelisting}{Listing}

\makeatother
\makeatletter
\@ifpackageloaded{caption}{}{\usepackage{caption}}
\@ifpackageloaded{subcaption}{}{\usepackage{subcaption}}
\makeatother

\newtheorem{theorem}{Theorem}
\newtheorem{lemma}{Lemma}
\newtheorem{corollary}{Corollary}
\newtheorem{proposition}{Proposition}

\ifLuaTeX
  \usepackage{selnolig}  
\fi
\usepackage[]{natbib}
\usepackage{bookmark}

\IfFileExists{xurl.sty}{\usepackage{xurl}}{} 
\hypersetup{
  pdftitle={Sharp Bounds for the LATE and the ATE with an Invalid Instrument},
  pdfauthor={Author 1; Author 2},
  pdfkeywords={3 to 6 keywords, that do not appear in the title},
  colorlinks=true,
  linkcolor={blue},
  filecolor={Maroon},
  citecolor={Blue},
  urlcolor={Blue},
  pdfcreator={LaTeX via pandoc}}

\newcommand{\anon}{1}

\begin{document}

\def\spacingset#1{\renewcommand{\baselinestretch}%
{#1}\small\normalsize} \spacingset{1}


\if1\anon
{
  \title{\bf Sensitivity Analysis for Instrumental Variables Under Joint Relaxations of Monotonicity and Independence}
  \author{Pedro Picchetti\thanks{ \textit{pedro.picchetti@uc.cl}}\hspace{.2cm}\\
    Institute of Economics, PUC-Chile}
  \maketitle
} \fi

\if0\anon
{
  \bigskip
  \bigskip
  \bigskip
  \begin{center}
    {\LARGE\bf Title}
\end{center}
  \medskip
} \fi

\bigskip
\begin{abstract}
In this paper I develop a breakdown frontier approach to assess the sensitivity of Local Average Treatment Effects (LATE) estimates to violations of monotonicity and independence of the instrument. I parametrize violations of independence using the concept of $c$-dependence from \cite{mastenpoirier2018} and allow for the share of defiers to be greater than zero but smaller than the share of compliers. I derive identified sets for the LATE and the Average Treatment Effect (ATE) in which the bounds are functions of these two sensitivity parameters. Using these bounds, I derive the breakdown frontier for the LATE, which is the weakest set of assumptions such that a conclusion regarding the LATE holds. I derive consistent sample analogue estimators for the breakdown frontiers and provide a valid bootstrap procedure for inference. Monte Carlo simulations show the desirable finite-sample properties of the estimators and an empirical application shows that the conclusions regarding the effect of family size on female labor force participation from \cite{angev} are highly sensitive to violations of independence and monotonicity.
\end{abstract}

\noindent%
\vfill

\newpage
\spacingset{1.8} 

\section{Introduction}

Instrumental variables (IV) techniques are among the most widely used empirical tools in social sciences. In the canonical IV setting, the causal effect of a binary treatment is identified by exploiting variations in a binary instrument in the form of the \cite{wald} estimand. Point identification is achieved if the instrument satisfies a set of assumptions. For instance, the instrumental variable must be independent from potential treatments and potential outcomes. Also, the instrument must affect treatment uptake in the same direction for all individuals, which is usually referred to as the monotonicity assumption.

If the instrument satisfies the independence and the monotonicity assumption, along additional assumptions, then the Wald estimand identifies the average effect of the treatment for compliers, the subpopulation of individuals whose treatment status mimics its assignment, which is called the Local Average Treatment Effect, or simply LATE \citep{imbensangrist}.

In the recent years, applied researchers have grown increasingly skeptical of IV methods \citep{cinelli}. The identifying assumptions are often unverifiable, and although in certain cases some assumptions are readily justified (for instance, the independence assumption in experimental studies with imperfect compliance), in most cases they are defended by appealing to context-specific knowledge.

In this paper I study what can be learned about treatment effects in IV settings under relaxations of independence and monotonicity and develop a breakdown frontier approach  for assessing the sensitivity of IV estimates. I focus in the case where the outcome is binary. I begin by deriving bounds for potential treatments and potential outcomes under a bounded dependence assumption called \textit{c-dependence} \citep{mastenpoirier2018}, which bounds the distance between the probability of being assigned to treatment given observed covariates and unobserved potential quantities and the probability of assignment given just the observed covariates.

I then use the bounds for potential quantities to derive identified sets for the causal effect of assignment (which I will refer to as the Intention-to-Treat, or simply ITT) and the LATE, under a share of defiers which is greater than zero, but always smaller than the share of compliers. I derive the conditions under which these identified sets are sharp. In the cases where these conditions do not hold, the identified sets still provide a valid outer region for the parameters of interest.

One can argue that once the identifying assumptions for IV settings are violated, the LATE is no longer an interesting causal parameter. Thus, I also derive the identified set for the Average Treatment Effect (ATE) and show how the bounds under violations of independence and monotonicity are connected to well known bounds for the ATE using IVs in the causal inference literature \citep{Balke01091997,Chen1015-7483R1}.

I use the bounds of the ITT and the LATE to construct breakdown frontiers for conclusions regarding causal effects. The breakdown frontier in this setting provides the largest combination of violations of independence and monotonicity under which a particular conclusion holds. For instance, suppose a researcher finds a positive point estimate for the LATE, but is skeptical towards the identifying assumptions. To provide evidence of robustness of the qualitative takeaways of its findings (for instance, that the effect is indeed positive), the researcher can use the breakdown frontier to show the combinations of violations under which one can still conclude that the LATE is greater than zero.

I propose nonparametric estimators for the bounds of causal effects and breakdown values, and derive their asymptotic properties using convergence results for \textit{Hadamard
directional differentiable} functions \citep{fangsantos}. Standard inference methods such as the nonparametric bootstrap are not consistent for the breakdown frontiers. I show, however, that valid uniform confidence bands can be estimated using the boostrap procedure for \textit{Hadamard
directional differentiable} functions in \cite{fangsantos} and the numerical estimator for the Hadamard derivative from \cite{HONG2018379}. Monte Carlo simulations show the desirable finite sample properties of the estimators and inference procedures.

For the empirical application, I revisit \cite{angev}, which studies the effects of family size on female employment using same-sex siblings as the instrument. The estimated breakdown frontier for the LATE shows that the qualitative takeaway from this study only holds under very small violations of the identifying assumptions. Therefore, the breakdown frontier approach suggests that the conclusions of the study are highly sensitive to violations of independence and monotonicity.

\textbf{Related Literature:} This paper relates broadly to three strands of the causal inference literature. First, it is connected to the literature on partial identification and sensitivity analysis in IV settings. Most papers in this literature focus on partial identification and sensitivity analysis under violations of independence and the exclusion restriction\citep{conley,wang18,mastenporirer21,cinelli}. There also papers that focus on identification and sensitivity analysis under violations of monotonicity \citep{tolerating,noack2026sensitivity}. In this paper, I consider both relaxations of independence and monotonicity.

Second, this paper relates to the literature on the identification of breakdown values, introduced by \cite{horwitzmanski}. My approach to inference follows closely the one introduced in \cite{mastenpoirier2020} as it also uses $c$-dependence to parametrize violations of independence. While most of the work in this literature focuses on missing data settings \citep{klinesantos} and selection on observables \citep{mastenpoirier2020}, this is one of the first papers studying inference for breakdown values in settings with non-compliance. In that sense, it is closely related to the work of \cite{noack2026sensitivity}, but under a different parametrization for violations of monotonicity. A desirable feature of the breakdown analysis in this paper is that the violations of independence and monotonicity are measured in the same unit, which makes the interpretation of the tradeoffs of violations displayed by the breakdown frontier particularly easy.

Finally, this paper is related to the literature on IV settings with binary outcomes, which dates back to the seminal work of \cite{heckman78}. While most prominent work on this literature focuses on the identification of the average structural functions \citep{vytlacilyildiz,shaikhvytlacil} or partial identification of Average Treatment Effects \citep{Balke01091997,Chen1015-7483R1,MACHADO2019522}, this paper considers both the partial identification of the LATE and the ATE.

\textbf{Outline of the paper:} The rest of the paper is organized as follows: Section 2 describes the framework and target parameters in the setting. Section 3 provides the partial identification of potential treatments and outcomes, and in Section 4 I derive the identified sets for the ITT and the LATE show the identification of the breakdown frontiers. I also derive the identified sets for the ATE. Section 5 introduces the estimators and their asymptotic properties, as well as the bootstrap procedure used for inference. Section 6 presents the Monte Carlo simulation studies. Section 7 presents the empirical application and Section 8 concludes. Appendix A contains the main proofs from the results in the paper, and Appendix B contains auxiliary lemmas.

\section{General Framework}

\subsection*{Setup}

Let $Z\in\left\{ 0,1\right\}$ denote a binary variable that indicates whether an individual was assigned to treatment ($Z=1$) or control ($Z=0$). In this setting, non-compliance is allowed, which means that not all individuals assigned to treatment will actually take the treatment and not all individuals assigned to control will remain untreated. Rather than determining treatment status, the assignment represents an encouragement (or discouragement) towards treatment.

Let $D\in\left\{ 0,1\right\}$ denote the actual treatment status. Define the potential treatment associated to assignment $z$ as $D(z)$. We observe the treatment status

\vspace{-10mm}

\begin{equation*}
    D=ZD(1)+(1-Z)D(0)
\end{equation*}

Let $Y\in\left\{ 0,1\right\}$ denote the observed binary outcome. The potential outcome associated to assignment $z$ is defined as $Y(D(z),z)$. At first, I allow potential outcomes depend arbitrarily on treatment and assignment. Observed and potential outcomes are related by

\vspace{-10mm}

\begin{equation*}
    Y=ZY(D(1),1)+(1-Z)Y(D(0),0)
\end{equation*}

Let $X\in\mathcal{S}(X)$ be a vector of observed covariates and $p_{z|x}=\mathbb{P}\left ( Z=z|X=x \right )$ be the observed propensity score for assignment. I maintain the following assumption regarding the joint distribution of $(D(z),Y(D(z),z),Z,X)$ throughout the paper:

\textbf{Assumption 1:} For each $z,z'\in\left\{ 0,1\right\}$ and $x\in\mathcal{S}(X)$:

\begin{enumerate}
    \item $\mathbb{P}\left ( D(z)=1|Z=z',X=x \right )\in \left ( 0,1 \right )$
    \item $\mathbb{P}\left ( Y(D(z),z)=1|Z=z',X=x \right )\in \left ( 0,1 \right )$
    \item $\mathbb{P}\left ( Y(D(z),z)=1,D(z)=1|Z=z',X=x \right )\in \left ( 0,1 \right )$
    \item $p_{z|x}>0$
\end{enumerate}

Assumptions 1.1 to 1.3 state that the support of potential quantities does not depend on the assignment. Assumption 1.4 states that all individuals can be assigned to treatment and control with probability greater than zero, and is usually referred to as the common support, or overlap assumption.

I also maintain the standard exclusion restriction assumption from IV settings, which imposes that assignment does not affect potential outcomes directly.

\textbf{Assumption 2:} For $(D(z),z)\in\left\{0,1\right\}^{2}$, $Y(D(z),z)=Y(D(z))$.

In order to identify treatment effects with instrumental variables, it is standard to assume that the instrument is independent of potential outcomes and potential treatments conditional on $X=x$. The goal of this identification analysis is to study what can be said about treatment effects when standard IV assumptions fail to hold. To do this, I replace these standard assumptions by a bounded dependence assumption, called \textit{c-dependence} \citep{mastenpoirier2018}:

\textbf{Definition ($c$-dependence):} Let $z\in\left\{0,1\right\}$ and $x\in\mathcal{S}(X)$. Let $c$ be a scalar between 0 and 1. $Z$ is conditionally $c$-dependent with $Y(D(z)),D(z)$ given $X=x$ if

\vspace{-10mm}

\begin{equation*}
\sup_{(y,d)\in\mathcal{S}_{x}((Y(D(z)),D(z)))}\left | \mathbb{P}\left ( Z=1|Y(D(z))=y,D(z)=d,X=x \right )-\mathbb{P}\left ( Z=1|X=x \right )\right |\leq c
\end{equation*}

\noindent where $\mathcal{S}_{x}((Y(D(z)),D(z)))$ is the support of $(Y(D(z)),D(z))$ conditional on $X=x$. Conditional $c$-dependence provides a parametrization of violations of independence which has a straightforward interpretation. The sensitivity parameter $c$ can be interpreted as the difference between the unobserved assignment probability and the observed propensity score in terms of probability units. When $c=0$, independence holds, and potential probabilities $\mathbb{P}\left(Y(D(z))=y|X=x\right)$ and $\mathbb{P}\left(D(z)=d|X=x\right)$ are point identified. Throughout this paper, $c$-dependence is assume to hold.

\textbf{Assumption 3:} $Z$ is $c$-dependent with $(Y(D(1)),D(1))$ given $X$ and $(Y(D(0)),D(0))$ given $X$.

Without further assumptions, individuals can be partitioned into four groups regarding how they respond to assignment: always-takers ($at$), never-takers ($nt$), compliers ($co$) and defiers ($def$). Let $\pi_{g|x}$ denote the proportion of individuals from group $g\in\left\{at,nt,co,def\right\}$ with covariates equal to $x$. The fundamental behavioral assumption in IV settings is the monotonicity assumption, which imposes that for all $x\in\mathcal{S}(X)$, $\pi_{def|x}=0$.

I relax this assumption to allow for the presence of defiers, but I restrict the proportion of defiers to be smaller than the share of compliers:

\textbf{Assumption 4:} For all $x\in\mathcal{S}(X)$, $\pi_{co|x}>\pi_{def|x}$.

Assumption 4 is analogous to Assumption 5 from \cite{tolerating}. If this assumption holds, then the the estimate for the first-stage is positive\footnote{See \cite{noack2026sensitivity} for a partial ID framework where defiance is allowed which relaxes this assumption.}. Thus, the sensitivity parameter $\pi_{def|x}$ can be seen as a measure of deviation between the proportion of compliers $\pi_{co|x}$ and the first-stage estimand $\mathbb{E}\left[D|Z=1,X=x\right]-\mathbb{E}\left[D|Z=0,X=x\right]$ in terms of probability units.

If assumption 3 holds with $c=0$, then potential quantities are identified. If Assumption 4 further holds with $\pi_{def|x}=0$ for all $x$, then we go back to the standard IV setting with binary outcomes, where the LATE is point identified by the Wald estimand, and the ATE is partially identified within the bounds provided by \cite{Balke01091997}.

\subsection*{Target Parameters}

In this paper, I focus on the partial Identification of the Local Average Treatment Effect for compliers, which is usually the target parameter in IV settings and the Average Treatment Effect (ATE), which is typically the causal parameter that researchers would ideally like to identify. Define $LATE_{g}=\mathbb{E}\left[Y(1)-Y(0)|g\right]$ as the local average treatment effect for group $g$, with $g\in\left\{at,nt,co,def\right\}$. We are thus, interested in the partial identification of $LATE_{co}$ and the identification of a breakdown frontier which can be used to assess the robustness of results from studies that employ IV methods.

Researchers often report a point estimate of the paramerer $LATE_{co}$ because they assume the parameter is point identified in their instrumental variable setting. However, it is often argued that $LATE_{co}$ is not necessarily a relevant parameter \citep{huber2017jae}. Moreover, once the instrument is not assumed to be independent from conditional quantities, nor it is assumed to be monotonic, then potential quantities for the sub-population of compliers are no longer point identified. Therefore, I also focus on the partial identification of the ATE.

The partial identification approach is built using the following steps. First, I derive the bounds for conditional potential joint probabilities $\mathbb{P}\left(Y(D(z))=y,D(z)=d|X=x\right)$. Then, I derive the bounds for the marginal probabilities of potential outcomes and potential treatments, $\mathbb{P}\left(Y(D(z))=y|X=x\right)$ and $\mathbb{P}\left(D(z)=d|X=x\right)$. Then, I derive bounds for the parameter $LATE_{co|x}$, which is the LATE for compliers conditional on $X=x$, and a breakdown frontier, and finally bounds for the conditional ATE. Unconditional quantities are partially identified by integrating the conditional bounds over the distribution of covariates.

\section{Partial Identification of Potential Probabilities}

I begin with the identified set for the joint probability of potential quantities. I begin with the joint probability of potential quantities. Under Assumptions 1 and 2, the results from Proposition 5 of \cite{mastenpoirier2018} can be readily adapted to the IV setting and the modified conditional $c$-dependence assumption. Let $p_{y,d|z,x}=\mathbb{P}\left(Y=y,D=d|Z=z,X=x\right)$:

\begin{proposition}
    Suppose Assumptions 1-3 hold. Then the sharp identified set for $\mathbb{P}\left(Y(D(z))=y,D(z)=d|X=x\right)$ is

    \vspace{-10mm}

    {\small \begin{equation*}
    \mathbb{P}\left ( Y(D(z))=y,D(z)=d|X=x \right )\in\left [\mathbb{P}_{c}^{l}\left ( Y(D(z))=y,D(z)=d|X=x \right ),\mathbb{P}_{c}^{u}\left ( Y(D(z))=y,D(z)=d|X=x \right )  \right ]
\end{equation*}}

where

\vspace{-10mm}

\begin{align*}
    &\mathbb{P}_{c}^{l}\left ( Y(D(z))=y,D(z)=d|X=x \right )=\\&\max\left\{\frac{p_{y,d|z,x}p_{z|x}}{p_{z|x}+c},\frac{p_{y,d|z,x}p_{z|x}-c}{p_{z|x}-c}\mathbf{1}\left ( p_{z|x}>c \right ),p_{y,d|z,x}p_{z|x}\right\},\\&\mathbb{P}_{c}^{u}\left ( Y(D(z))=y,D(z)=d|X=x \right )\\&=\min\left\{\frac{p_{y,d|z,x}p_{z|x}}{p_{z|x}-c}\mathbf{1}\left ( p_{z|x}>c \right )+\mathbf{1}\left ( p_{z|x}\leq c \right ),\frac{p_{y,d|z,x}p_{z|x}+c}{p_{z|x}+c},p_{y,d|z,x}p_{z|x}+(1-p_{z|x})\right\}
\end{align*}

\end{proposition}

The notation introduced in Proposition 1 for the bounds on the joint probability of potential quantities illustrates the fact the bounds are functions of the sensitivity parameter $c$. When $c=0$ the joint probability is point identified by the conditional probability $\mathbb{P}\left(Y=y,D=d|Z=z,X=x\right)$,and as $c$ increases the identified set becomes larger until we reach the worst-case identified set.

The bounds from Proposition 1 can be combined using the Law of Total Probabilities to obtain bounds for the marginal probabilities of potential quantities. I begin with the bounds for potential outcomes:

\begin{proposition}
    Suppose Assumptions 1-3 hold. Then the identified set for $\mathbb{P}\left(Y(D(z))=y|X=x\right)$ is 

    \vspace{-10mm}

    \begin{equation*}
        \mathbb{P}\left ( Y(D(z))=y|X=x \right )\in\left [\mathbb{P}_{c}^{l}\left ( Y(D(z))=y|X=x \right ),\mathbb{P}_{c}^{u}\left ( Y(D(z))=y|X=x \right )  \right ]
    \end{equation*}

    \noindent where

    \vspace{-10mm}

\begin{align*}
    &\mathbb{P}_{c}^{l}\left ( Y(D(z))=y|X=x \right )\\&=\max\left\{ \mathbb{P}^{l}_{c}\left(Y(D(z))=y,D(z)=1|X=x\right)+\mathbb{P}^{l}_{c}\left(Y(D(z))=y,D(z)=0|X=x\right),p_{y|z,x}p_{z|x}\right\},\\&\mathbb{P}_{c}^{u}\left ( Y(D(z))=y|X=x \right )\\&=\min\left\{\mathbb{P}^{u}_{c}\left(Y(D(z))=y,D(z)=1|X=x\right)+\mathbb{P}^{u}_{c}\left(Y(D(z))=y,D(z)=0|X=x\right),p_{y|z,x}p_{z|x}+(1-p_{z|x})\right\}
\end{align*}

Moreover, if $p_{y,d|z,x}<1/2$ and $c\in\left ( 0,\min\left\{ p_{z|x}(1-2p_{y,d|z,x}),\frac{p_{z|x}(1-p_{z|x})(1-p_{y,d|z,x})}{p_{y,d|z,x}p_{z|x}+1-p_{z|x}}\right\} \right )$, then the identified set is sharp.
\end{proposition}

Proposition 2 shows that the bounds for potential outcomes can be obtained by combining the bounds for join potential quantities, and the additional conditions under which the identified set is sharp. Essentially, the additional conditions imply that the upper bound for joint potential probabilities simplifies to 
$\mathbb{P}^{u}_{c}\left(Y(D(z))=y,D(z)=d|X=x\right)=\frac{p_{y,d|z,x}p_{z|x}}{p_{z|x}-c}$, and that the lower bound simplifies to $\mathbb{P}^{l}_{c}\left(Y(D(z))=y,D(z)=d|X=x\right)=\frac{p_{y,d|z,x}p_{z|x}}{p_{z|x}+c}$. If this conditions do not hold, the identified set still provides a valid outer region.

The results from Proposition 2 can be easily adapted to obtain bounds for potential treatments:

\begin{proposition}
    Suppose Assumptions 1-3 hold. Then the identified set for $\mathbb{P}\left(D(z)=d|X=x\right)$ is 

    \vspace{-10mm}

    \begin{equation*}
        \mathbb{P}\left ( D(z)=d|X=x \right )\in\left [\mathbb{P}_{c}^{l}\left ( D(z)=d|X=x \right ),\mathbb{P}_{c}^{u}\left ( D(z)=d|X=x \right )  \right ]
    \end{equation*}

    \noindent where

    \vspace{-10mm}

    \begin{align*}
    &\mathbb{P}_{c}^{l}\left ( D(z)=d|X=x \right )\\&=\max\left\{ \mathbb{P}^{l}_{c}\left(Y(D(z))=1,D(z)=d|X=x\right)+\mathbb{P}^{l}_{c}\left(Y(D(z))=0,D(z)=d|X=x\right),p_{d|z,x}p_{z|x}\right\},\\&\mathbb{P}_{c}^{u}\left ( D(z)=d|X=x \right )\\&=\min\left\{\mathbb{P}^{u}_{c}\left(Y(D(z))=1,D(z)=d|X=x\right)+\mathbb{P}^{u}_{c}\left(Y(D(z))=0,D(z)=d|X=x\right),p_{d|z,x}p_{z|x}+(1-p_{z|x})\right\}
\end{align*}

Moreover, if $p_{y,d|z,x}<1/2$ and $c\in\left ( 0,\min\left\{ p_{z|x}(1-2p_{y,d|z,x}),\frac{p_{z|x}(1-p_{z|x})(1-p_{y,d|z,x})}{p_{y,d|z,x}p_{z|x}+1-p_{z|x}}\right\} \right )$, then the identified set is sharp.
\end{proposition}

Bounds for unconditional potential probabilities are obtained by integrating the bounds of conditional probabilities over the distribution of covariates. The bounds derived in this section are the building blocks for the partial identification of treatment effects which is presented in the next section.

\section{Partial Identification of Treatment Effects}

\subsection{Partial Identification of $LATE_{co}$}

I begin deriving bounds for the average treatment effects of the sub-population of compliers. In the standard IV setting with covariates, the parameter $LATE_{co|x}$ is partially identified by the conditional Wald estimand:

\vspace{-10mm}

{\small \begin{equation*}
    \frac{\mathbb{E}\left[Y|Z=1,X=x\right]-\mathbb{E}\left[Y|Z=0,X=x\right]}{\mathbb{E}\left[D|Z=1,X=x\right]-\mathbb{E}\left[D|Z=0,X=x\right]}=\frac{\mathbb{E}\left[Y(1)-Y(0)|co,X=x\right]\pi_{co|x}}{\pi_{co|x}}=\mathbb{E}\left[Y(1)-Y(0)|co,X=x\right]
\end{equation*}}

The numerator of the Wald estimand, usually referred to as the reduced form estimand, identifies the causal effect of assignment, $\mathbb{E}\left[Y(D(1))-Y(D(0))|X=x\right]$, which is equal to the treatment effect for compliers multiplied by the share of compliers in the standard IV setting. This parameter is often called the Intention-to-Treat effect (I will refer to its conditional as $ITT_{x}$ and its unconditional version as $ITT$). The ITT is rarely the parameter of interest in IV settings, but it carries important information regarding the LATE for compliers. For instance, the paramater $LATE_{co}$ has the same sign as the ITT if monotonicity holds, or if the share of compliers is greater than the share of defiers. The next proposition provides the bounds for the ITT as functions of the sensitivity parameters $c$ and $\pi_{def|x}$, as well as the conditions under which these bounds are sharp.

\begin{proposition}
    Suppose Assumptions 1-4 hold. Then $ITT_{x}\in\left[ITT^{l}(c,\pi_{def|x}),ITT^{u}(c,\pi_{def|x})\right]$, where

    \vspace{-10mm}

    \begin{align*}
        &ITT^{u}(c,\pi_{def|x})=\min\left\{\mathbb{P}_{c}^{u}\left(Y(D(1))=1|X=x\right)-\mathbb{P}_{c}^{l}\left(Y(D(0))=1|X=x\right)+\pi_{def|x},1\right\},\\&ITT^{l}(c,\pi_{def|x})=\max\left\{\mathbb{P}_{c}^{l}\left(Y(D(1))=1|X=x\right)-\mathbb{P}_{c}^{u}\left(Y(D(0))=1|X=x\right)-\pi_{def|x},-1\right\}
    \end{align*}

    Moreover, if

    \vspace{-25mm}

    \begin{align*}
        &\\&(i)\max\left\{0, \mathbb{P}\left ( D(0)=1|X=x \right )-\mathbb{P}\left ( D(1)=1|X=x \right ) \right\}\\&\leq \pi_{def|x}\\&\leq \min\left\{ \mathbb{P}\left ( D(0)=1|X=x \right ),\mathbb{P}\left ( D(1)=0|X=x \right )\right\}\\&(ii)\ p_{y,d|z,x}<1/2\\&(iii)\ c\in\left ( 0,\min\left\{ p_{z|x}(1-2p_{y,d|z,x}),\frac{p_{z|x}(1-p_{z|x})(1-p_{y,d|z,x})}{p_{y,d|z,x}p_{z|x}+1-p_{z|x}}\right\} \right )\\&(iv)\  \mathbb{P}^{u}_{c}\left(Y(D(0))=1|X=x\right)-\mathbb{P}^{l}_{c}\left(Y(D(1))=1|X=x\right)-1\\&<\pi_{def|x}\\&<1-\left(\mathbb{P}^{u}_{c}\left(Y(D(1))=1|X=x\right)-\mathbb{P}^{l}_{c}\left(Y(D(0))=1|X=x\right)\right)
    \end{align*}
    
   \noindent for all $\left(y,d\right)\in\left\{0,1\right\}^{2}$ and $x\in\mathcal{S}(X)$, then the identified set is sharp.
\end{proposition}

Proposition 4 provides bounds for the conditional ITT. The bounds for the unconditional ITT are obtained by integrating the conditional bounds over the distribution of covariates. Under additional assumptions, the bound is sharp. These assumptions restrict the share of complier to lie within the Fréchet-feasible interval (i), the values which joint potential probabilities can take (ii), the values which the sensitivity parameter $c$ can take (iii) and the share of defiers to lie in an interval in which the truncations of the bounds are not active (iv). If these assumptions fail to hold, the bounds still provide a valid outer region for the ITT.

Proposition 4 provides bounds for the ITT. Once bounds for the share of compliers $\pi_{co|x}$ are obtained, one can derive the identified set for $LATE_{co|x}$:

\begin{proposition}
    Suppose Assumptions 1-4 hold. Then the identified set for $LATE_{co|x}$ is $\left[LATE^{u}_{co|x}(c,\pi_{def|x}),LATE^{l}_{co|x}(c,\pi_{def|x})\right]$, where

    \vspace{-10mm}

    \begin{align*}
    &LATE_{co|x}^{UB}(c,\pi_{def|x})=\min\left\{\frac{ITT_{x}^{u}(c,\pi_{def|x})}{\pi_{co|x}^{l}(c,\pi_{def|x})},1\right\},\\&LATE_{co|x}^{LB}(c,\pi_{def|x})=\max\left\{\frac{ITT^{l}_{x}(c,\pi_{def|x})}{\pi_{co|x}^{u}(c,\pi_{def|x})},-1\right\}
\end{align*}

with

\vspace{-10mm}

\begin{align*}
    &\pi_{co|x}^{u}(c,\pi_{def|x})=\min\left\{\mathbb{P}_{c}^{u}\left(D(1)=1|X=x\right)-\mathbb{P}^{l}\left(D(0)=1|X=x\right)+\pi_{def|x},1\right\},\\&\pi_{co|x}^{l}(c,\pi_{def|x})=\max\left\{ \mathbb{P}_{c}^{l}\left(D(1)=1|X=x\right)-\mathbb{P}^{u}\left(D(0)=1|X=x\right)+\pi_{def|x},0\right\}
\end{align*}

Moreover, if 

\vspace{-25mm}

    \begin{align*}
        &\\&(i)\max\left\{0, \mathbb{P}\left ( D(0)=1|X=x \right )-\mathbb{P}\left ( D(1)=1|X=x \right ) \right\}\\&\leq \pi_{def|x}\\&\leq \min\left\{ \mathbb{P}\left ( D(0)=1|X=x \right ),\mathbb{P}\left ( D(1)=0|X=x \right )\right\}\\&(ii)\ 
       c=0\\&(iii)\ \pi_{def|x}<1-\left(\mathbb{P}\left(D=1|Z=1,X=x\right)-\mathbb{P}\left(D=1|Z=0,X=x\right)\right)
    \end{align*}
    
   \noindent for all $\left(y,d\right)\in\left\{0,1\right\}^{2}$ and $x\in\mathcal{S}(X)$, then the identified set is sharp.
    
\end{proposition}

Proposition 5 provides the bounds for the LATE of compliers. In general, the bounds will not be sharp, since under violations of independence (Assumption 4 holds with $c>0$), the upper bound of the conditional ITT and the lower bound of the conditional share of compliers (and vice-versa) cannot be attained simultaneously while satisfying Assumptions 1-4. Nevertheless, in the case where independence holds and the share of compliers is such that it satisfies the Frechet inequalities and the bounds for the share of compliers are not the worst-case bounds, the identified set is sharp.

\subsection{Breakdown Frontier}

In this section I provide a breakdown frontier approach to assess the robustness of ITT and LATE estimates to violations of independence and monotonicity. I focus on the breakdown frontiers for the conclusions that $ITT\geq \mu_{1}$ and $LATE_{co}\geq \mu_{2}$. Choosing $\mu_{1}$ and $\mu_{2}$ equal to 0, for instance, provides us the breakdown analysis for the conclusion that the treatment has a positive effect. 

When deriving the breakdown frontier, it is important to consider the same share of defiers across all values of covariates ($\pi_{def|x}=\pi_{def}$) in order to go obtain unconditional bounds. First, consider all the values of $c$ and $\pi_{def}$ under which the conclusion holds. This sets are called the robust regions and are defined for the ITT and the LATE, respectively, as

\vspace{-10mm}

\begin{align*}
    &RR_{ITT}^{}(\mu_{1})=\left\{\left(c,\pi_{def}\right)\in\left[0,1\right]^{2}:ITT^{l}(c,\pi_{def})\geq \mu_{1}\right\},\\&RR_{LATE}^{}(\mu_{2})=\left\{\left(c,\pi_{def}\right)\in\left[0,1\right]^{2}:LATE_{co}^{l}(c,\pi_{def})\geq \mu_{2}\right\}
\end{align*}

Robust regions are simply combinations of $\left(c,\pi_{def}\right)$ which respectively deliver identified sets for the ITT and $LATE_{co}$ that contain the values $\mu_{1}$ and $\mu_{2}$. The breakdown frontiers are the sets $\left(c,\pi_{def}\right)$ in the boundary of the robust region for given conclusions. The breakdown frontiers are

\vspace{-10mm}

\begin{align*}
    &BF_{ITT}^{}(\mu_{1})=\left\{\left(c,\pi_{def}\right)\in\left[0,1\right]^{2}:ITT^{l}(c,\pi_{def})= \mu_{1}\right\},\\&BF_{LATE}^{}(\mu_{2})=\left\{\left(c,\pi_{def}\right)\in\left[0,1\right]^{2}:LATE_{co}^{l}(c,\pi_{def})= \mu_{2}\right\}
\end{align*}

Note that, in the cases where the bounds for the ITT and the LATE are not sharp, the breakdown frontiers operate as a conservative sufficient-robustness frontier rather than the exact combination of breakdown values.

Solving for $\pi_{def}$ in the equations $ITT^{l}(c,\pi_{def})=\mu_{1}$ and $LATE_{c}^{l}(c,\pi_{def})=\mu_{2}$ yields

\vspace{-10mm}

\begin{align*}
    &bf_{ITT}^{}(c,\mu_{1})=\int\mathbb{P}^{l}_{c}\left(Y(D(1))=1|X=x\right)dF_{X}(x)-\int\mathbb{P}^{u}_{c}\left(Y(D(0))=1|X=x\right)dF_{X}(x)-\mu_{1},\\&bf_{LATE}^{}(c,\mu_{2})=\frac{\int\mathbb{P}^{l}_{c}\left(Y(D(1))=1|X=x\right)dF_{X}(x)-\int\mathbb{P}^{u}_{c}\left(Y(D(0))=1|X=x\right)dF_{X}(x)}{1+\mu_{2}}\\&-\frac{\mu_{2}\left\{\int\mathbb{P}^{u}_{c}\left(D(1)=1|X=x\right)dF_{X}(x)-\int\mathbb{P}^{l}_{c}\left(D(0)=1|X=x\right)dF_{X}(x)\right\}}{1+\mu_{2}}
\end{align*}

Therefore, we obtain the following analytical expressions for the breakdown frontiers:

\vspace{-10mm}

\begin{align*}
    &BF_{ITT}^{}(c,\mu_{1})=\min\left\{\max\left\{bf_{ITT}^{}(c,\mu_{1}),0\right\},1\right\},\\&BF_{LATE}^{}(c,\mu_{2})=\min\left\{\max\left\{bf_{LATE}^{}(c,\mu_{2}),0\right\},1\right\}
\end{align*}

The frontiers provide the largest relaxations $c$ and $\pi_{def}$ under which predetermined conclusions regarding the ITT and the LATE hold. The shape of the frontier allows us to analyze the trade-off between the two types of relaxations considered when drawing conclusions regarding the target parameters. A desirable feature of this approach is that the sensitivity parameters $c$ and $\pi_{def}$ are measured in the same unit. Although this is not necessary, it certainly can be helpful.

Note that when we are interested in assessing the conclusion regarding the sign of treatment effect we can always use the breakdown frontiers for the ITT, as $bf_{LATE}^{}(c,0)=bf_{ITT}^{}(c,0)$. Next, I provide a simple numerical illustration of the bounds of the treatment effects and the breakdown frontier approach.

\subsubsection{Numerical Illustration}

I consider a simple DGP with a single covariate $X$ where $x\in\left\{0,1\right\}$. The instrument is assigned according to a Bernoulli distribution with parameter $p=0.6$. The covariate is distributed according to a Bernoulli distribution with parameter $p=0.5$. See Appendix C for the entire characterization of the DGP. Potential outcomes and treatments are defined in a way such that for $x\in\left\{0,1\right\}$, we have

\vspace{-10mm}

\begin{align*}
    &\mathbb{P}\left(Y(D(1))=1|X=x\right)-\mathbb{P}\left(Y(D(0))=1|X=x\right)=0.25,\\&\frac{\mathbb{P}\left(Y(D(1))=1|X=x\right)-\mathbb{P}\left(Y(D(0))=1|X=x\right)}{\mathbb{P}\left(D(1)=1|X=x\right)-\mathbb{P}\left(D(0)=1|X=x\right)}=0.5
\end{align*}

Therefore, in the absence of violations of the identifying assumptions in IV settings, the ITT is equal to 0.25 and the LATE is equal to 0.5 under this DGP. To analyze the sensitivity to violations of independence and monotonicity, Figures 1 shows the identified sets for the LATE under different shares of defiers.

\begin{figure}[htbp]
\caption{Identified Sets for the LATE}
\centering
\includegraphics[width=0.45\textwidth]{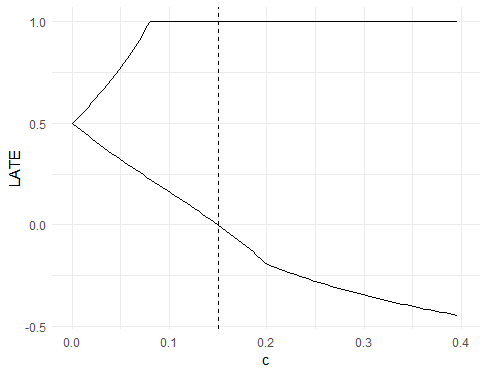}
\hfill
\includegraphics[width=0.45\textwidth]{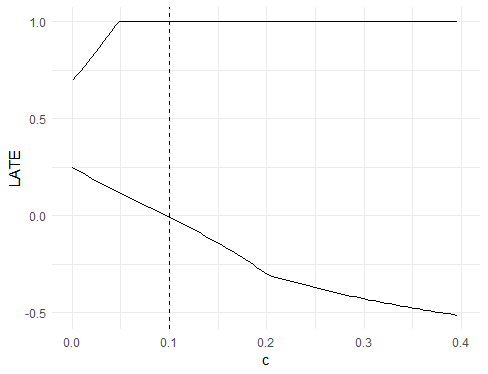}\\
    \scriptsize \noindent \textit{Note:} Left: Identified set for the LATE as a function of $c$ setting $\pi_{def}=0$. Right: Identified set for the LATE as a function of $c$ setting $\pi_{def}=0.1$ The vertical lines represent the values of $c$ under which the bounds become uninformative.
\end{figure}

The plot on the left of Figure 1 shows the identified set for the LATE under the monotonicity assumption ($\pi_{def|x}=0$ for all $x$). When $c=0$, the identified set collapses to 0.5, which is the value which would be point identified in the absence of any violation. As $c$ increases, the set becomes less informative. The vertical dotted line marks the largest violation $c$ under which the identified set does not contain 0. That is, in the absence of defiers, we can conclude that the LATE is positive under violations of independence for all sensitivity parameters $c\leq0.15$.

The plot on the right shows the identified set when defiance is allowed (I set $\pi_{def|x}=0.1$ for all $x$). Note that in this case, the LATE is no longer point identified when $c=0$. The vertical dotted line is moved to the left, and shows that we can conclude that the LATE is positive for violation parameters $c\leq 0.1$.

The plots with the identified sets under different shares of defiers illustrate the tradeoff between the magnitude of the violations when assessing the robustness of a given conclusion regarding the LATE. If we want to conclude that the LATE is positive, 
we can allow for smaller deviations from independence as we allow for larger shares of defiers.

The breakdown frontier format captures the tradeoffs between these violations. Figure 2 shows the breakdown frontiers for two conclusions regarding the LATE.

\begin{figure}[htbp]
\caption{Breakdown Frontiers}
\centering
\includegraphics[width=0.45\textwidth]{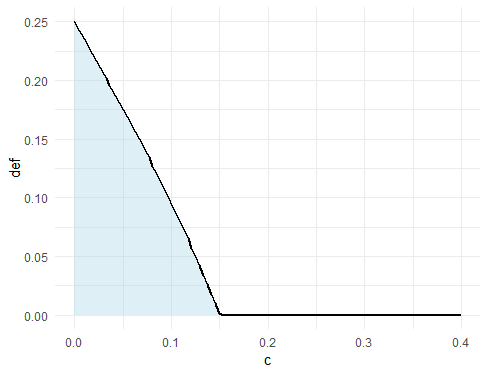}
\hfill
\includegraphics[width=0.45\textwidth]{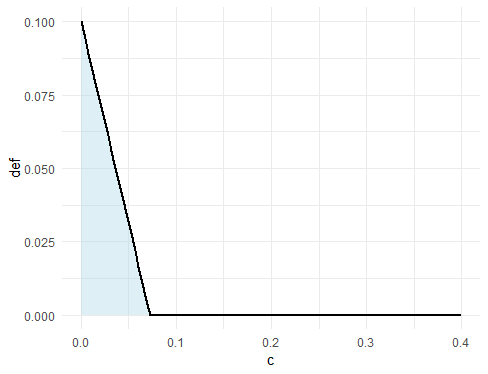}\\
    \scriptsize \noindent \textit{Note:} Left: Breakdown frontier for the conclusion that the LATE is greater than zero. Right: Breakdown frontier for the conclusion that the LATE is greater than 0.25. The blue areas are the robust regions for the conclusions.
\end{figure}

The plot on the left of Figure 2 provides the breakdown frontier for the conclusion that $LATE_{co}^{l}(c,\pi_{def})\geq 0$. The area painted in blue represents the robust region for the conclusion that the LATE is positive, and the black line denotes the breakdown frontier. The breakdown frontier shows that if we are willing to assume independence, then the share of defiers can be as great as 0.25 and the conclusion that the LATE is positive still holds. If we are willing to assume monotonicity, then the observed and unobserved propensity scores can differ by up to 0.15 probability units and the conclusion still holds.

The plot on the right shows the robust region and the breakdown frontier for the conclusion that $LATE^{l}(c,\pi_{def})\geq 0.25$, which is half of the value that is point identified under the standard assumptions. Note that the robust region is smaller that the one for the conclusion that the LATE is positive, and smaller violations of monotonicity are admitted in order for the conclusion to hold. If we are willing to assume that independence holds, then we can allow for a share of defiers no greater than 0.1. If we are willing to assume monotonicity, then the observed and unobserved propensity scores can differ by up to 0.075 probability units and the conclusion still holds.

\subsection{Partial Identification of the ATE}

Researchers usually report the LATE in IV settings because that is the causal parameter that is point identified under the standard IV assumptions \citep{imbensangrist}. However, whether or not the LATE is a relevant parameter depends on the empirical context \citep{huber2017jae,Chen1015-7483R1}. Researchers are typically interested in the Average Treatment Effect (ATE), which is the most general average causal parameter. Moreover, once the standard IV assumptions are violated and potential quantities are no longer point identified for the group of compliers, it might be of interest to analyze what can be learned about the ATE.

The ATE is a parameter that is not point identified in standard IV settings, as the quantities $\mathbb{P}\left(Y(0)=1|at,X=x\right)$ and $\mathbb{P}\left(Y(1)=1|nt,X=x\right)$ cannot be point identified from the data without further assumptions. If violations of monotonicity are allowed, further potential outcomes cannot be point identified. If violations of independence is also allowed, then none of the potential quantities are identified. The next proposition shows what are the bounds for the ATE under violations of monotonicity and independence.

\begin{proposition}
    Suppose Assumptions 1-4 hold. Then, $ATE_{x}\in\left[ATE^{l}_{x}(c,\pi_{def|x}),ATE^{u}_{x}(c,\pi_{def|x})\right]$, where

    \vspace{-10mm}

    \begin{align*}
        &ATE_{x}^{u}(c,\pi_{def|x})=\mathbb{P}^{u}\left(Y(1)=1|X=x\right)-\mathbb{P}^{l}\left(Y(0)=1|X=x\right),\\&ATE^{l}_{x}(c,\pi_{def|x})=\mathbb{P}^{l}\left(Y(1)=1|X=x\right)-\mathbb{P}^{u}\left(Y(0)=1|X=x\right)
    \end{align*}

    \noindent with

    \vspace{-10mm}

    \begin{equation*}
\begin{aligned}
\mathbb{P}^{u}\left ( Y(1)=1|X=x \right )
= \min\Bigg\{&
\mathbb{P}^{u}_{c}\left(Y(D(1))=1,D(1)=1|X=x\right)  +\mathbb{P}^{u}_{c}\left(Y(D(0))=1,D(0)=1|X=x\right)  \\
&+\mathbb{P}^{u}_{c}\left(D(1)=0|X=x\right)-\pi_{def|x},\\&\mathbb{P}\left (Y=1|D=1,X=x  \right )\mathbb{P}\left (D=1|X=x  \right )+(1-\mathbb{P}\left (D=1|X=x  \right ))
\Bigg\},
\end{aligned}
\end{equation*}

\vspace{-10mm}

\begin{equation*}
\begin{aligned}
\mathbb{P}^{l}\left ( Y(1)=1|X=x \right )
\\= \max\Bigg\{&
\mathbb{P}^{l}_{c}\left(Y(D(1))=1,D(1)=1|X=x\right)  +\mathbb{P}^{l}_{c}\left(Y(D(0))=1,D(0)=1|X=x\right)  \\
&-\min\left\{\mathbb{P}^{u}_{c}\left(Y(D(1))=1,D(1)=1|X=x\right),\mathbb{P}^{u}_{c}\left(Y(D(0))=1,D(0)=1|X=x\right)\right\},\\&\mathbb{P}\left (Y=1|D=1,X=x  \right )\mathbb{P}\left (D=1|X=x  \right )
\Bigg\},
\end{aligned}
\end{equation*}

\vspace{-10mm}

\begin{equation*}
\begin{aligned}
\mathbb{P}^{u}\left ( Y(0)=1|X=x \right )
= \min\Bigg\{&
\mathbb{P}^{u}_{c}\left(Y(D(1))=1,D(1)=0|X=x\right)  +\mathbb{P}^{u}_{c}\left(Y(D(0))=1,D(0)=0|X=x\right)  \\
&+\mathbb{P}^{u}_{c}\left(D(0)=1|X=x\right)-\pi_{def|x},\\&\mathbb{P}\left (Y=1|D=0,X=x  \right )\mathbb{P}\left (D=0|X=x  \right )+(1-\mathbb{P}\left (D=0|X=x  \right ))
\Bigg\},
\end{aligned}
\end{equation*}

\vspace{-10mm}

\begin{equation*}
\begin{aligned}
\mathbb{P}^{l}\left ( Y(0)=1|X=x \right )
\\=\max\Bigg\{&
\mathbb{P}^{l}_{c}\left(Y(D(1))=1,D(1)=0|X=x\right)  +\mathbb{P}^{l}_{c}\left(Y(D(0))=1,D(0)=0|X=x\right)  \\
&-\min\left\{\mathbb{P}^{u}_{c}\left(Y(D(1))=1,D(1)=0|X=x\right),\mathbb{P}^{u}_{c}\left(Y(D(0))=1,D(0)=0|X=x\right)\right\},\\&\mathbb{P}\left (Y=1|D=0,X=x  \right )\mathbb{P}\left (D=0|X=x  \right )
\Bigg\}
\end{aligned}
\end{equation*}

Moreover, if

    \vspace{-25mm}

    \begin{align*}
        &\\&(i)\max\left\{0, \mathbb{P}\left ( D(0)=1|X=x \right )-\mathbb{P}\left ( D(1)=1|X=x \right ) \right\}\\&\leq \pi_{def|x}\\&\leq \min\left\{ \mathbb{P}\left ( D(0)=1|X=x \right ),\mathbb{P}\left ( D(1)=0|X=x \right )\right\}\\&(ii)\ p_{y,d|z,x}<1/2\\&(iii)\ c\in\left ( 0,\min\left\{ p_{z|x}(1-2p_{y,d|z,x}),\frac{p_{z|x}(1-p_{z|x})(1-p_{y,d|z,x})}{p_{y,d|z,x}p_{z|x}+1-p_{z|x}}\right\} \right )\\&(iv)\  \mathbb{P}^{u}_{c}\left(Y(D(0))=1|X=x\right)-\mathbb{P}^{l}_{c}\left(Y(D(1))=1|X=x\right)-1\\&<\pi_{def|x}\\&<1-\left(\mathbb{P}^{u}_{c}\left(Y(D(1))=1|X=x\right)-\mathbb{P}^{l}_{c}\left(Y(D(0))=1|X=x\right)\right)
    \end{align*}
    
   \noindent for all $\left(y,d\right)\in\left\{0,1\right\}^{2}$ and $x\in\mathcal{S}(X)$, then the identified set is sharp.
\end{proposition}

The bounds in Proposition 6 can be directly connected to the existing bounds for the ATE in the IV literature. The next corollary shows that the bounds from proposition 6 are equivalent to the bounds from \cite{Balke01091997} and \cite{Chen1015-7483R1} in the absence of violations.

\begin{corollary}
    Suppose Assumptions 1-4 hold. Furthermore, suppose that Assumption 3 holds with $c=0$ and Assumption 4 holds with $\pi_{def|x}=0$. Then, the bounds for $ATE_{x}$ become

    \vspace{-10mm}

    \begin{align*}
        &ATE^{u}_{x}=\mathbb{P}\left(Y=1,D=1|Z=1,X=x\right)-\mathbb{P}\left(Y=1,D=0|Z=0,X=x\right)+\mathbb{P}\left(D=0|Z=1,X=x\right),\\&ATE^{l}_{x}=\mathbb{P}\left(Y=1,D=1|Z=1,X=x\right)-\mathbb{P}\left(Y=1,D=0|Z=0,X=x\right)-\mathbb{P}\left(D=1|Z=0,X=x\right)
    \end{align*}
\end{corollary}

\section{Estimation and Inference}

In this section, I study estimation and inference of the bounds for the LATE and the breakdown frontiers for the LATE and ITT defined in Section 4.1. The bounds and the breakdown frontier are known functionals of conditional probabilities of treatments and outcomes given assignments and covariates, and the conditional probabilities of assignments given covariates. Hence, I propose nonparametric sample analogue estimators for the bounds and the breakdown frontier.

First I assume a random sample of data is available for the researcher:

\textbf{Assumption 5:} The random variables $\left\{Y_{i},D_{i},Z_{i},X_{i}\right\}_{i=1}^{N}$ are independently and identically distributed according to the distribution of $\left(Y,D,Z,X\right)$.

Furthermore, assume that the support of the vector of covariates is discrete:

\textbf{Assumption 6:} The support of $X$ is discrete and finite. Let $\mathcal{S}(X)=\left\{x_{1},...,x_{K}\right\}$. 

Next, I invoke an assumption which is an important regularity condition for the derivation of the asymptotic properties of the estimator.

\textbf{Assumption 7:} For all $x\in\mathcal{S}(X)$, we have $c<\min\left\{p_{1|x},p_{0|x}\right\}$.

Assumption 7 is necessary for the proposed bounds to be sharp, but is also key for asymptotics. The asymptotic results are obtained using a delta method for directionally differentiable functionals. Under assumption 7, the indicator functions inside the min and max operators that determines the bounds disappear, and therefore, there are no Dirac delta functions in the analytical expression.

I begin with the asymptotic properties of the bounds for the LATE and its breakdown frontier.

\subsection{LATE and Breakdown Frontier}

The parameters of interest defined in Section 4.1 are functionals of the parameters $p_{y,d|z,x}=\mathbb{P}\left(Y=y,D=d|Z=z,X=x\right)$, $p_{z|x}=\mathbb{P}\left(Z=z|X=x\right)$ and $q_{x}=\mathbb{P}\left(X=x\right)$. Let

\begin{align*}
    &\widehat{p}_{y,d|z,x}=\frac{\frac{1}{N}\sum_{i=1}^{N}\mathbf{1}\left\{Y_{i}=y,D_{i}=d\right\}\mathbf{1}\left\{Z_{i}=z,X_{i}=x\right\}}{\frac{1}{N}\sum_{i=1}^{N}\mathbf{1}\left\{Z_{i}=z,X_{i}=x\right\}},\\&\widehat{p}_{z|x}=\frac{\frac{1}{N}\sum_{i=1}^{N}\mathbf{1}\left\{Z_{i}=z,X_{i}=x\right\}}{\frac{1}{N}\sum_{i=1}^{N}\mathbf{1}\left\{X_{i}=x\right\}},\\&\widehat{q}_{x}=\frac{1}{N}\sum_{i=1}^{N}\mathbf{1}\left\{X_{i}=x\right\}
\end{align*}

\noindent denote the sample analog estimators of these probabilities. In Lemma 1 of Appendix B, I show that the estimators of these quantities converge uniformly to a Gaussian process at a $\sqrt{N}$-rate.

The bounds in Propositions 1-6 are functionals evaluated at $p_{y,d|z,x}$, $p_{z|x}$ and $q_{x}$. The bounds are estimated by these functionals evaluated at the sample analogue estimators. If these functionals are \textit{Hadamard
directional differentiable}, then $\sqrt{N}$-convergence in distribution of the sample analogue estimators will carry over to the functionals by the delta method. 

I use the functional delta method for
Hadamard directionally differentiable mappings \citep{fangsantos} to show convergence in distribution of the estimators. Convergence is usually to a non-Gaussian limiting process. Thus, analytical asymptotic bands are challenging to obtain. I follow \cite{mastenpoirier2020} and propose a bootstrap procedure to obtain asymptotically valid uniform confidence bands for the breakdown frontier and the estimators for the bounds.

Consider the bounds from Proposition 1. Under Assumptions 1-7, we estimate them by

\vspace{-10mm}

\begin{align*}
    &\widehat{\mathbb{P}}^{u}_{c}\left(Y(D(z))=y,D(z)=d|X=x\right)=\min\left\{\frac{\widehat{p}_{y,d|z,x}\widehat{p}_{z|x}}{\widehat{p}_{z|x}-c},\frac{\widehat{p}_{y,d|z,x}\widehat{p}_{z|x}+c}{\widehat{p}_{z|x}+c},\widehat{p}_{y,d|z,x}\widehat{p}_{z|x}+(1-\widehat{p}_{z|x})\right\},\\&\widehat{\mathbb{P}}^{l}_{c}\left(Y(D(z))=y,D(z)=d|X=x\right)=\min\left\{\frac{\widehat{p}_{y,d|z,x}\widehat{p}_{z|x}}{\widehat{p}_{z|x}+c},\frac{\widehat{p}_{y,d|z,x}\widehat{p}_{z|x}-c}{\widehat{p}_{z|x}-c},\widehat{p}_{y,d|z,x}\widehat{p}_{z|x}\right\}
\end{align*}

The estimators perform poorly when $c$ is close to $p_{z|x}$. Assumption 7 ensures that $c$ is bounded away from $p_{z|x}$. The estimators for the bounds of potential treatments are analogous. In Lemmas 3 and 4 from Appendix B I show that these estimators converge in
distribution to a nonstandard distribution.

For the main results in this section I establish convergence uniformly over $c\in\mathcal{C}$, where $\mathcal{C}$ is a finite grid $\mathcal{C}\subset\left[0,\min\left\{p_{1|x},p_{0|x}\right\}\right]$ for all $x\in\mathcal{S}(X)$. Therefore, the asymptotic results are valid for values of $c$ which satisfy Assumption 7.

Next, consider the bounds for the conditional ITT introduced in Proposition 4. We estimate them by

\vspace{-10mm}

\begin{align*}
    &\widehat{ITT}_{x}^{u}(c,\pi_{def|x})=\min\left\{\widehat{\mathbb{P}}^{u}_{c}\left(Y(D(1))=1|X=x\right)-\widehat{\mathbb{P}}^{l}_{c}\left(Y(D(0))=1|X=x\right)+\pi_{def|x},1\right\},\\&\widehat{ITT}_{x}^{l}(c,\pi_{def|x})=\max\left\{\widehat{\mathbb{P}}^{l}_{c}\left(Y(D(1))=1|X=x\right)-\widehat{\mathbb{P}}^{u}_{c}\left(Y(D(0))=1|X=x\right)-\pi_{def|x},-1\right\}
\end{align*}

The unconditional bounds are estimated  by integrating over the empirical distribution of the covariates $X$. Let

\vspace{-10mm}

\begin{equation*}
    \widehat{ITT}^{u}(c,\pi_{def})=\frac{1}{N}\sum_{i=1}^{N}\widehat{ITT}^{u}_{X_{i}}(c,\pi_{def})\ \text{and}\ \widehat{ITT}^{l}(c,\pi)=\frac{1}{N}\sum_{i=1}^{N}\widehat{ITT}^{l}_{X_{i}}(c,\pi_{def})
\end{equation*}

In Lemma 5 of Appendix B, I show that these estimators for the ITT bounds converge weakly to a Gaussian element.

Now, consider the estimation for the breakdown frontier for the conclusion that the ITT is above a certain threshold $\mu_{1}$. Although the ITT is not the usual parameter of interest in IV settings, the breakdown frontier for the conclusion that the ITT is greater than zero coincides with the breakdown frontier for the conclusion that the LATE is greater than zero, so it is interesting to analyze its asymptotic properties.

Denote the breakdown frontier for the conclusion that $ITT\geq\mu_{1}$ by

\vspace{-10mm}

\begin{equation*}
    \widehat{BF}_{ITT}(c,\mu_{1})=\min\left\{\max\left\{\widehat{bf}_{ITT}(c,\mu_{1}),0\right\},1\right\}
\end{equation*}

\noindent where

\vspace{-10mm}

\begin{equation*}
    \widehat{bf}_{ITT}=\frac{1}{N}\sum_{i=1}^{N}\widehat{\mathbb{P}}^{l}_{c}\left(Y(D(1))=1|X_{i}\right)-\frac{1}{N}\sum_{i=1}^{N}\widehat{\mathbb{P}}^{u}_{c}\left(Y(D(0))=1|X_{i}\right)-\mu_{1}
\end{equation*}

I show that the estimator for the breakdown frontier of the ITT converges in distribution.

\begin{theorem}
    Suppose Assumptions 1-7 hold and that $c\in\mathcal{C}$ for some finite grid $\mathcal{C}\subset \overline{C}\in\left(0,\min\left\{p_{1|x},p_{0|x}\right\}\right)$. Let $\mathcal{M}\subset\left[-1,1\right]$ be a finite grid of points. Then, 

    \begin{equation*}
        \sqrt{N}\left(\widehat{BF}_{ITT}(c,\mu)-BF_{ITT}(c,\mu)\right)\xrightarrow[d]{}\textbf{Z}_{BF_{ITT}}(c,\mu)
    \end{equation*}

    \noindent a tight random element of $l^{\infty}\left(\mathcal{C}\times\mathcal{M}\right)$.
\end{theorem}

Now, consider the bounds for the conditional LATE introduced in Proposition 5. They are obtained by combining the bounds for the ITT with the bounds for the share of compliers. We estimate the bounds for the share of compliers by

\vspace{-10mm}

\begin{align*}
    &\widehat{\pi}^{u}_{co|x}(c,\pi_{def|x})=\min\left\{\widehat{\mathbb{P}}_{c}^{u}\left(D(1)=1|X=x\right)-\widehat{\mathbb{P}}_{c}^{l}\left(D(0)=1|X=x\right)+\pi_{def|x},1\right\},\\&\widehat{\pi}^{l}_{co|x}(c,\pi_{def|x})=\max\left\{\widehat{\mathbb{P}}_{c}^{l}\left(D(1)=1|X=x\right)-\widehat{\mathbb{P}}_{c}^{u}\left(D(0)=1|X=x\right)+\pi_{def|x},0\right\}
\end{align*}

The unconditional bounds are estimated  by integrating over the empirical distribution of the covariates $X$. Let

\vspace{-10mm}

\begin{equation*}
    \widehat{\pi}^{u}_{co}(c,\pi_{def})=\frac{1}{N}\sum_{i=1}^{N}\widehat{\pi}^{u}_{co|X_{i}}(c,\pi_{def})\ \text{and}\ \widehat{\pi}_{co}^{l}(c,\pi_{def})=\frac{1}{N}\sum_{i=1}^{N}\widehat{\pi}^{l}_{co|X_{i}}(c,\pi_{def})
\end{equation*}

In Lemma 6 of Appendix B, I show that these estimators converge weakly to a Gaussian element. The estimators for the bounds of the LATE are obtained by combining the bounds of the ITT and the share of compliers:

\vspace{-10mm}

\begin{align*}
    &\widehat{LATE}^{u}(c,\pi_{def})=\min\left\{\frac{\widehat{ITT}^{u}(c,\pi_{def})}{\widehat{\pi}^{l}_{co}(c,\pi_{def})},1\right\},\\&\widehat{LATE}^{l}(c,\pi_{def})=\max\left\{\frac{\widehat{ITT}^{l}(c,\pi_{def})}{\widehat{\pi}^{u}_{co}(c,\pi_{def})},-1\right\}
\end{align*}

In Lemma 7 of Appendix B, I show that these estimators converge weakly to a Gaussian element. The estimator for the breakdown frontier for the conclusion that $LATE\geq \mu_{2}$ is

\vspace{-10mm}

\begin{equation*}
    \widehat{BF}_{LATE}(c,\mu_{2})=\min\left\{\max\left\{\widehat{bf}_{LATE}(c,\mu_{2}),0\right\},1\right\}
\end{equation*}

\noindent where

\vspace{-10mm}

\begin{align*}
    &\widehat{bf}_{LATE}=\frac{1}{N}\sum_{i=1}^{N}\frac{\widehat{\mathbb{P}}^{l}_{c}\left(Y(D(1))=1|X_{i}\right)-\widehat{\mathbb{P}}^{u}_{c}\left(Y(D(0))=1|X_{i}\right)}{1+\mu_{2}}\\&-\frac{\mu_{2}\left\{\widehat{\mathbb{P}}^{u}_{c}\left(D(1)=1|X_{i}\right)-\widehat{\mathbb{P}}^{l}_{c}\left(D(0)=1|X_{i}\right)\right\}}{1+\mu_{2}}
\end{align*}

I show that the estimator for the breakdown frontier of the LATE converges in distribution.

\begin{theorem}
    Suppose Assumptions 1-7 hold and that $c\in\mathcal{C}$ for some finite grid $\mathcal{C}\subset \overline{C}\in\left(0,\min\left\{p_{1|x},p_{0|x}\right\}\right)$. Let $\mathcal{M}\subset\left[-1,1\right]$ be a finite grid of points. Then, 

    \begin{equation*}
        \sqrt{N}\left(\widehat{BF}_{LATE}(c,\mu)-BF_{LATE}(c,\mu)\right)\xrightarrow[d]{}\textbf{Z}_{BF_{LATE}}(c,\mu)
    \end{equation*}

    \noindent a tight random element of $l^{\infty}\left(\mathcal{C}\times\mathcal{M}\right)$.
\end{theorem}

The results in this section essentially follow from the $\sqrt{N}$-convergence rate of the sample analogue estimators to a Gaussian process and by sequential applications of the Delta Method for Hadamard directionally differentiable functions.

\subsection{Bootstrap Inference}

The limiting processes of the estimators presented in this Section are non-Gaussian, so relying on analytical estimates of quantiles of functionals of these processes would be challenging. In order to overcome these challenges I use the bootstrap procedure from \cite{mastenpoirier2020}. The bootstrap procedure is subsequently used to construct uniform confidence bands for the breakdown frontiers.

Let $W_{i}=\left(Y_{i},D_{i},Z_{i},X_{i}\right)$ and $W^{N}=\left\{W_{1},W_{2},...,W_{N}\right\}$. Let $\theta_{0}$ denote a parameter of interest and $\widehat{\theta}$ be an estimator of $\theta_{0}$ based on $W^{N}$. Define $\textbf{A}_{N}^{*}=\sqrt{N}\left(\widehat{\theta}^{*}-\widehat{\theta}\right)$, where $\widehat{\theta}^{*}$ is a draw from the nonparametric bootstrap distribution of $\widehat{\theta}$.

I focus on

\vspace{-10mm}

\begin{equation*}
    \theta_{0}=\begin{pmatrix}
p_{y,d|z,x}
 \\p_{z|x}
 \\q_{x}
\end{pmatrix}\ \text{and}\ \widehat{\theta}=\begin{pmatrix}
\widehat{p}_{y,d|z,x}
 \\\widehat{p}_{z|x}
 \\\widehat{q}_{x}
\end{pmatrix}
\end{equation*}

Let $\textbf{Z}_{1}$ denote the limiting distribution of $\sqrt{N}\left(\widehat{\theta}-\theta_{0}\right)$, which is defined in Lemma 1 of Appendix B. It is well known that $\textbf{A}_{N}^{*}$ converges weakly to $\textbf{Z}_{1}$. The parameters of interest are functionals $\phi$ of $\theta_{0}$. For Hadamard differentiable functions, the nonparametric bootstrap is valid \citep{fangsantos}. However, when parameters are only Hadamard directionally differentiable, which is the case for the bounds of the ITT and LATE, and the breakdown frontiers, the nonparametric bootstrap is not consistent.

To construct a consistent bootstrap distribution, I use the bootstrap procedure from \cite{fangsantos}, which relies on a consistent estimator $\widehat{\phi^{'}}_{\theta_{0}}$ of the Hadamard derivative at $\theta_{0}$. These estimates can be obtained by using the numerical derivative estimator proposed by \cite{HONG2018379}, which is

\vspace{-10mm}

\begin{equation*}
    \widehat{\phi^{'}}_{\theta_{0}}\left(\sqrt{N}(\widehat{\theta}^{*}-\widehat{\theta})\right)=\frac{\phi\left(\widehat{\theta}+\varepsilon_{N}\sqrt{N}\left(\widehat{\theta}^{*}-\widehat{\theta}\right)\right)-\phi\left(\widehat{\theta}\right)}{\varepsilon_{N}}
\end{equation*}

\noindent and is computed across the bootstrap estimates $\widehat{\theta}^{*}$ Under the constraints $\varepsilon_{N}\rightarrow0$ and $\sqrt{N}\varepsilon_{N}\rightarrow\infty$ and additional regularity conditions, this numerical derivative bootstrap procedure is consistent (Li and Hong, 2018).

I use this bootstrap procedure  construct uniform confidence bands for the breakdown frontiers. I focus on one-sided lower uniform confidence bands. I am looking for a lower bound function $\widehat{LB}(c)$ such that

\vspace{-10mm}

\begin{align*}
    &\lim_{N\rightarrow\infty}\mathbb{P}\left(\widehat{LB}(c)\leq BF_{ITT}(c,\mu)\ \text{for all}\ c\in\left[0,1\right]\right)=1-\alpha
\end{align*}

I consider bands of the form

\vspace{-10mm}

\begin{equation*}
    \widehat{LB}(c)=\widehat{BF}(c,\mu)-\widehat{z}_{1-\alpha}\frac{\sigma(c)}{\sqrt{N}}
\end{equation*}

\noindent where $\widehat{z}_{1-\alpha}$ is a scalar and $\sigma(.)$ is a known function. Note that under Assumptions 1-7, the estimators for the breakdown frontiers can be written as $\widehat{BF}(c,\mu)=\left[\phi(\widehat{\theta})\right](c)$, where $\phi$ is Hadamard directionally differentiable. If we further assume that $\varepsilon_{N}\rightarrow0$ and $\sqrt{N}\varepsilon_{N}\rightarrow\infty$, then the conditions in Proposition 2 from \cite{mastenpoirier2020} hold, and the estimator

\vspace{-10mm}

\begin{equation*}
    \widehat{z}_{1-\alpha}=\inf\left\{ z\in\mathbb{R}:\mathbb{P}\left(\sup_{c\in\mathcal{C}}\frac{\left [ \widehat{\phi^{'}}_{\theta_{0}}\left ( \sqrt{N}\left ( \widehat{\theta}^{*}-\widehat{\theta} \right ) \right ) \right ](c,\mu)}{\sigma(c)}\leq z|W^{N}\right)\geq 1-\alpha\right\}
\end{equation*}

is consistent for $z_{1-\alpha}$, the $1-\alpha$ quantile of the cdf of

\vspace{-10mm}

\begin{equation*}
    \sup_{c\in\mathcal{C}}\frac{\textbf{Z}_{BF}(c,\mu)}{\sigma(c)}
\end{equation*}

Note that this holds for the estimators of both breakdown frontiers. It follows that the proposed lower bands are valid uniformly on the grid $\mathcal{C}$. In the next section, I study the finite-sample properties of the estimation and inference procedures for breakdown frontiers.

\section{Monte Carlo Simulations}

In this section I study the finite sample performance of the estimation and inference procedures proposed in Section 5. I consider the same DGP from the numerical illustration in Section 4.2.1, which implies a joint distribution for $\left(Y,D,Z,X\right)$ from which I draw independently. 

I consider two sample sizes, $N=1000$ and $N=2000$. For each sample size, I conduct 500 Monte Carlo simulations. For each exercise, I compute the estimated breakdown frontier and a 95\% lower bootstrap uniform confidence band. In all simulations, I set $\varepsilon_{N}=2/\sqrt{N}$, which is the choice of $\varepsilon_{N}$ which shows the best finite-sample coverage in \cite{MastenPoirier2020_supplement}. I estimate the breakdown frontier over a finite grid of points $c$. In the simulation, I use 100 values of $c$ equally spaced between 0 and 0.15 both in all simulations and bootstrap procedures.

\begin{figure}[htbp]
\caption{Sampling Distribution of the Breakdown Frontier Estimator}
\centering
\includegraphics[width=0.45\textwidth]{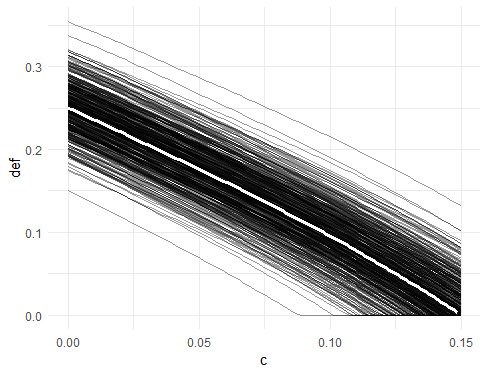}
\hfill
\includegraphics[width=0.45\textwidth]{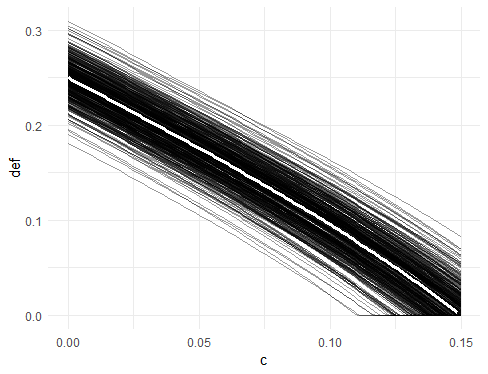}
\\
    \scriptsize \noindent \textit{Note:} Left: N = 1.000. Right: N = 2000. These plots show the sampling distribution of our
breakdown frontier estimator by gathering the point estimates of the breakdown frontier across
all Monte Carlo simulations into one plot. The true breakdown frontier is shown on top in white.
\end{figure}

Figure 3 shows the shows the sampling distribution of the breakdown frontier estimator for the conclusion that the LATE is greater than zero. The first thing that shows out is that, as implied by the consistency result in Section 5, the distribution of the estimator becomes tighter around the true frontier as the sample size increases. Second, the sampling distribution looks fairly symmetric around the true frontier. This contrasts with the findings of \cite{MastenPoirier2020_supplement}, which find that the estimator for the breakdown frontier of Distributional Treatment Effects is biased downwards. The difference might arise due to several factors: we consider different target parameters and different sensitivity parameters for the relaxation of the identifying assumptions, which inevitably leads to different functional forms for the breakdown frontiers. Nevertheless, the fact that the estimator for the breakdown frontier of the LATE is symmetric around the true frontier is a desirable feature which is does not hold generally for breakdown approach settings.

\begin{figure}[t]
    \centering
    \caption{Finite-Sample Bias of the Breakdown Frontier Estimator}
    \includegraphics{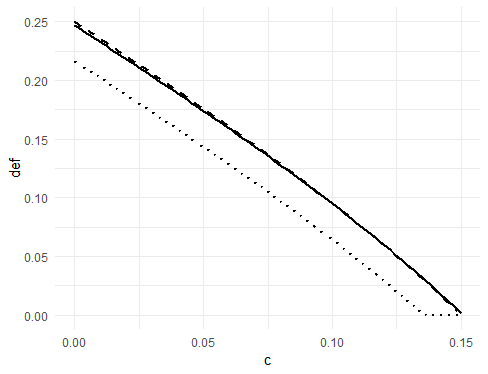}
    \label{fig:enter-label}\\
    \scriptsize \noindent \textit{Note:} This plot shows the finite-sample bias of the breakdown frontier estimator. The solid line is the true frontier, the dashed line the estimated finite sample mean of the frontier estimates and the dotted line the estimated finite sample mean of the 95\% lower confidence bands.
\end{figure}

Figure 4 shows the true breakdown frontier as the solid line, and the sample mean of breakdown estimates across the Monte Carlo simulations with $N=1.000$ as the dashed line. The two lines are pretty much overlapped, which shows that the finite-sample bias of the estimator for the frontier is very small across all considered values of $c$. The dotted line below represents the sample mean of the lower confidence band with nominal coverage $1-\alpha=0.95$.

Overall the results of the Monte Carlo exercise show desirable finite-sample properties of the estimator for the breakdown frontier. A pervasive concern when conducting inference procedures in IV settings is the so-called weak instrument problem. Although there are several bootstrap procedures that improve inference in settings with weak instruments where the standard assumptions hold, it is unclear how to improve the bootstrap for nondifferentiable functions. I leave this analysis for future work.

\section{Empirical Application}

In this section, I use the estimators from Section 5 to perform the breakdown analysis for the results regarding family size and female employment in \cite{angev}, using data from the US Census Public Use Microsamples married mothers aged 21–35 in 1980 with at least 2 children and oldest child less than 18.

In this setting, the dependent variable is and indicator for women who did not work for pay in 1979. Treatment is an indicator for women having three or more children, and the instrument is an indicator for women whose first two children have the same sex. The authors control for age, age at the first birth, race and sex of the first child as covariates.

Two concerns regarding the assumptions that lead to point identification of the LATE in this setting arise. The first, regards violations of monotonicity. The assumption holds if all parents in the sample have weak preferences towards mixed-sibling compositions. Although there is evidence that more families with two same-sex siblings have a higher probability of third birth than families with two siblings with mixed composition, this does not guarantee that there are no families which prefer same-sex siblings over mixed compositions. The second concern comes from the independence assumption. Genetic conditions which determine fertility outcomes can be correlated to economic outcomes \citep{farb}, which would lead to violations of independence. Under the light of this concerns, the \cite{angev} setting seems to be well suited for the breakdown analysis approach.

To begin the sensitivity analysis, I use selection on observables to to calibrate the beliefs regarding the amount of selection on unobservables. I take the approach from \cite{altonji} and \cite{mastenpoirier2018}. I partition the vector of covariates $X$ as $(X_{k},X_{-k})$, where $X_{k}$ is the $k$-th component and $X_{-k}$ is a vector with remaining components. The measures used to calibrate the beliefs regarding deviations from independence are

\vspace{-10mm}

\begin{equation*}
    \overline{c}_{k}=\sup_{x_{-k}}\sup_{x_{k}}\left | \mathbb{P}\left ( Z=1|X=(x_{k},x_{-k}) \right )-\mathbb{P}\left ( Z=1|X_{-k}=x_{-k} \right )\right |
\end{equation*}

In the data, the largest value obtained form $\overline{c}_{k}$ is associated to to the indicator for women whose first child is a man, which was estimated to be $\overline{c}_{1st\ sex}=0.011$. 

Using this result as a reference for the breakdown analysis, a robust result would have a breakdown frontier which admits values of $c$ above $\overline{c}_{1st\ sex}$. 

To calibrate the beliefs regarding violations of monotonicity, I follow \cite{tolerating}, which uses a survey from Peru in which women were asked about their ideal sex composition for their children. In the survey, 1.8\% of the respondents had three children or more and declared that ideal sex sibship composition would have been two boys and no girl, or no boy and two girls. Thus, one can argue that these women seem to have been induced to having a third
child because their first two children were a boy and a girl. Using this result as a reference, a robust result would have a breakdown frontier which admits a share of defiers greater than 0.018.

\begin{figure}[t]
    \centering
    \caption{Breakdown Frontier for the Effect of Family Size on Unemployment}
    \includegraphics{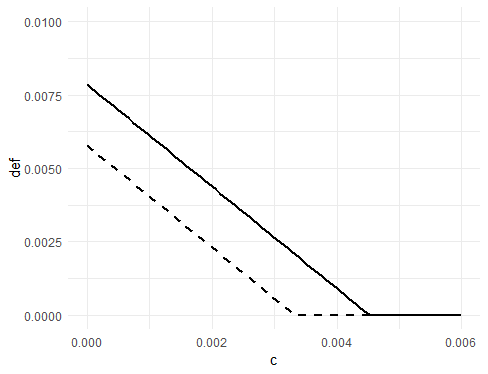}
    \label{fig:enter-label}\\
    \scriptsize \noindent \textit{Note:} Estimated breakdown frontier (solid line) for the conclusion that the effect of family size on unemployment is greater than zero. The dashed line is the 95\% lower confidence band.
\end{figure}

I estimate the breakdown frontier using 50 values of $c$, equally spaced between 0 and 0.1. For the construction of the lower confidence bands, I draw 999 boostrap samples from the data and set the tuning parameter to $\varepsilon_{N}=\frac{2}{\sqrt{N}}$. The implementation algorithm is similar to the one in Section 4 of \cite{mastenpoirier2020}, although it is much simpler since the considered outcome is binary.

Figure 5 shows the estimated breakdown frontier for the conclusion that the the effects of family size on employment is negative. The solid line is the estimated breakdown frontier, and the dashed line is the lower confidence band at the $1-\alpha=0.95$ level.

One can think of this the breakdown frontier as the frontier for the conclusion that the qualitative takeaways from \cite{angev} hold. The plot shows that when independence holds ($c=0$) the maximum share of defiers uner which the qualitative takeaways hold is 0.008, which lies below the baseline share of defiers implied the by the Peruvian survey. When monotonicity holds ($def=0$) the largest admissible difference between the observable and unobservable propensity scores is around 0.004 probability units, which lies below the baseline violation of 0.011 implied by the calibration based on selection on observables.

Overall, the results from this breakdown analysis suggest that the conclusion that effect of family size on employment is negative is not robust to violations of independence or monotonicity of the same-sex siblings instrument. The results align with the findings of \cite{noack2026sensitivity} which shows that small violations of monotonicity lead to uninformative results in this setting, and also add to the discussion that small deviations from independence also lead to uninformative results.

\section{Conclusion}

In this paper, I provide a breakdown frontier approach to sensitivity analysis in Instrumental Variables settings. I study the partial identification of the LATE under parametrizations of violations of independence and monotonicity. The bounds for the LATE are used to derived breakdown frontiers, the weakest set of assumptions such that a particular conclusion of interest
holds. Also, I derive identified set for the ATE under violations of independence and monotonicity given the fact that when the population of compliers is not point-identified, the LATE is no longer such a relevant parameter.

I propose sample analogue estimators and uniform confidence bands for the breakdown frontiers. Monte Carlo simulations show that the estimator exhibits desirable finite-sample properties. 

Finally, I use the proposed breakdown frontier approach to revisit the results from \cite{angev}, and find that the conclusions regarding the effect of family size on unemployment are highly sensitive to violations of independence and monotonicity.

\bibliography{bibliography}

\section*{Appendix A}

\subsection*{Proof of Proposition 1}

\subsubsection*{Validity}

The proof is the same as the one from Proposition 5 in \cite{mastenpoirier2018}, adapted to the version of conditional $c$-dependence where the probability of assignment is conditional on both potential outcomes and treatments.

\subsubsection*{Sharpness}

To show sharpness of the interior, with exhibit two joint distributions $\left(Y(D(z)),D(z),Z,X\right)$ consistent with the data and Assumptions 1-3. The first one yields the element $\mathbb{P}^{u}_{c}\left(Y(D(z))=y,D(z)=d|X=x\right)$ and the second one yields the element $\mathbb{P}^{l}_{c}\left(Y(D(z))=y,D(z)=d|X=x\right)$. If both bounds are attainable from DGPs which are consistent with the data and the assumptions, then all points in the identified set can be obtained by mixtures of these DGPs, and sharpness follows.

Since the distribution of $\left(Z,X\right)$ is observed, we need to specify a distribution for $(Y(D(z)),D(z))|Z,X$. We always observe $\mathbb{P}\left(Y(D(z))=y,D(z)=d|Z=z,X=x\right)=p_{y,d|z,x}$. Hence, we only need to specify a distribution for $\mathbb{P}\left(Y(D(z))=y,D(z)=d|Z=1-z,X=x\right)$.

I begin by specifying a value of $\mathbb{P}\left(Y(D(z))=y,D(z)=d|Z=1-z,X=x\right)$ such that

\begin{enumerate}
    \item $\mathbb{P}\left(Y(D(z))=y,D(z)=d|X=x\right)=\mathbb{P}^{u}_{c}\left(Y(D(z))=y,D(z)=d|X=x\right)$
    \item $\mathbb{P}\left(Y(D(z))=y,D(z)=d|Z=1-z,X=x\right)\in\left(0,1\right)$.
    \item Conditional $c$-dependence is satisfied.
\end{enumerate}

\textbf{Proof of 1:} Choose

\vspace{-10mm}

\begin{equation*}
    \mathbb{P}\left(Y(D(z))=y,D(z)=d|Z=1-z,X=x\right)=\frac{\mathbb{P}^{u}_{c}\left(Y(D(z))=y,D(z)=d|X=x\right)-p_{y,d|z,x}p_{z|x}}{1-p_{z|x}}
\end{equation*}

It follows that 

\vspace{-10mm}

\begin{align*}
    &\mathbb{P}\left(Y(D(z))=y,D(z)=d|X=x\right)\\&=p_{y,d|z,x}p_{z|x}+\mathbb{P}\left(Y(D(z))=y,D(z)=d|Z=1-z,X=x\right)(1-p_{z|x})\\&=\mathbb{P}^{u}_{c}\left(Y(D(z))=y,D(z)=d|X=x\right)
\end{align*}

\textbf{Proof of 2:}

Note that $\mathbb{P}^{u}_{c}\left(Y(D(z))=y,D(z)=d|X=x\right)\in\left(0,1\right)$ and that $p_{y,d|z,x}p_{z|x}\in\left(0,1\right)$. Moreover, note that $\mathbb{P}^{u}_{c}\left(Y(D(z))=y,D(z)=d|X=x\right)>p_{y,d|z,x}p_{z|x}$. And therefore, it follows that $\mathbb{P}\left(Y(D(z))=y,D(z)=d|Z=1-z,X=x\right)\in\left(0,1\right)$.

\textbf{Proof of 3:}

Conditional $c$-dependence implies that for all $\left(y,d,z\right)$, 

\vspace{-10mm}

\begin{equation*}
    \mathbb{P}\left(Z=z|Y(D(z))=y,D(z)=d,X=x\right)\in\left[p_{z|x}-c,p_{z|x}+c\right]
\end{equation*}

Using Bayes' rule, write $\mathbb{P}\left(Z=z|Y(D(z))=y,D(z)=d,X=x\right)$ as

\vspace{-10mm}

\begin{equation*}
    \mathbb{P}\left(Z=z|Y(D(z))=y,D(z)=d,X=x\right)=\frac{p_{z|x}p_{y,d|z,x}}{\mathbb{P}\left(Y(D(z))=y,D(z)|X=x\right)}
\end{equation*}

Decomposing the denominator using the Law of Total Probabilities, we find that in order for $c$-dependence to hold, it must be the case that $\mathbb{P}\left(Y(D(z))=y,D(z)|Z=1-z,X=x\right)$ lies in the interval

\vspace{-10mm}

\begin{equation*}
    \left[\frac{p_{y,d|z,x}p_{z|x}(1-(p_{z|x}+c))}{(p_{z|x}+c)(1-p_{z|x})},\frac{p_{y,d|z,x}p_{z|x}(1-(p_{z|x}-c))}{(p_{z|x}-c)(1-p_{z|x})}\right]
\end{equation*}

Note that

\vspace{-10mm}

\begin{align*}
    &\mathbb{P}\left(Y(D(z))=y,D(z)=d|Z=1-z,X=x\right)=\frac{\mathbb{P}^{u}_{c}\left(Y(D(z))=y,D(z)=d|X=x\right)-p_{y,d|z,x}p_{z|x}}{1-p_{z|x}}\\&\leq \frac{\frac{p_{y,d|z,x}p_{z|x}}{p_{z|x}-c}-p_{y,d|z,x}p_{z|x}}{1-p_{z|x}}=\frac{p_{y,d|z,x}p_{z|x}(1-(p_{z|x}-c))}{(p_{z|x}-c)(1-p_{z|x})}
\end{align*}

Also, note that 

\vspace{-10mm}

\begin{align*}
    &\mathbb{P}\left(Y(D(z))=y,D(z)=d|Z=1-z,X=x\right)=\frac{\mathbb{P}^{u}_{c}\left(Y(D(z))=y,D(z)=d|X=x\right)-p_{y,d|z,x}p_{z|x}}{1-p_{z|x}}\\&\geq \frac{\mathbb{P}^{l}_{c}\left(Y(D(z))=y,D(z)=d|X=x\right)-p_{y,d|z,x}p_{z|x}}{1-p_{z|x}}\\&\geq \frac{\frac{p_{y,d|z,x}p_{z|x}}{p_{z|x}+c}-p_{y,d|z,x}p_{z|x}}{1-p_{z|x}}=\frac{p_{y,d|z,x}p_{z|x}(1-(p_{z|x}+c))}{(p_{z|x}+c)(1-p_{z|x})}
\end{align*}

And thus, Assumption 3 holds. Now, I specify a value of $\mathbb{P}\left(Y(D(z))=y,D(z)=d|Z=1-z,X=x\right)$ such that

\begin{enumerate}
    \item $\mathbb{P}\left(Y(D(z))=y,D(z)=d|X=x\right)=\mathbb{P}^{l}_{c}\left(Y(D(z))=y,D(z)=d|X=x\right)$
    \item $\mathbb{P}\left(Y(D(z))=y,D(z)=d|Z=1-z,X=x\right)\in\left(0,1\right)$.
    \item Conditional $c$-dependence is satisfied.
\end{enumerate}

\textbf{Proof of 1:} Choose

\vspace{-10mm}

\begin{equation*}
    \mathbb{P}\left(Y(D(z))=y,D(z)=d|Z=1-z,X=x\right)=\frac{\mathbb{P}^{l}_{c}\left(Y(D(z))=y,D(z)=d|X=x\right)-p_{y,d|z,x}p_{z|x}}{1-p_{z|x}}
\end{equation*}

It follows that 

\vspace{-10mm}

\begin{align*}
    &\mathbb{P}\left(Y(D(z))=y,D(z)=d|X=x\right)\\&=p_{y,d|z,x}p_{z|x}+\mathbb{P}\left(Y(D(z))=y,D(z)=d|Z=1-z,X=x\right)(1-p_{z|x})\\&=\mathbb{P}^{l}_{c}\left(Y(D(z))=y,D(z)=d|X=x\right)
\end{align*}

\textbf{Proof of 2:}

Since both $\mathbb{P}^{l}_{c}\left(Y(D(z))=y,D(z)=d|X=x\right)$ and $p_{y,d|z,x}p_{z|x}$ lie between zero and one, it follows that $\mathbb{P}\left(Y(D(z))=y,D(z)=d|Z=1-z,X=x\right)\leq 1$. Also, we have

\vspace{-10mm}

\begin{align*}
    &\frac{\mathbb{P}^{l}_{c}\left(Y(D(z))=y,D(z)=d|X=x\right)-p_{y,d|z,x}p_{z|x}}{1-p_{z|x}}\geq \frac{p_{y,d|z,x}p_{z|x}}{1-p_{z|x}}\left(\frac{1}{\min\left\{p_{z|x}+c,1\right\}}-1\right)\\&\geq 0
\end{align*}

\textbf{Proof of 3:}

Follows from the proof for the upper bound. Therefore, there are DGPs consistent with the data and Assumptions 1-3 which attain the upper and the lower bound, from which sharpness follows.

\subsection*{Proof of Proposition 2}

\subsubsection*{Validity}

By the Law fo Total Probabilities, we have

\vspace{-10mm}

\begin{equation*}
    \mathbb{P}\left(Y(D(z))=y|X=x\right)=\mathbb{P}\left(Y(D(z))=y,D(z)=1|X=x\right)+\mathbb{P}\left(Y(D(z))=y,D(z)=0|X=x\right)
\end{equation*}

Hence, it follows that

\vspace{-10mm}

\begin{equation*}
    \mathbb{P}\left(Y(D(z))=y|X=x\right)\leq\mathbb{P}_{c}^{u}\left(Y(D(z))=y,D(z)=1|X=x\right)+\mathbb{P}_{c}^{u}\left(Y(D(z))=y,D(z)=0|X=x\right)
\end{equation*}

and

\vspace{-10mm}

\begin{equation*}
    \mathbb{P}\left(Y(D(z))=y|X=x\right)\geq\mathbb{P}_{c}^{l}\left(Y(D(z))=y,D(z)=1|X=x\right)+\mathbb{P}_{c}^{l}\left(Y(D(z))=y,D(z)=0|X=x\right)
\end{equation*}

Also, note that

\vspace{-10mm}

\begin{equation*}
    \mathbb{P}\left(Y(D(z))=y|X=x\right)=p_{y|z,x}p_{z|x}+\mathbb{P}\left(Y(D(z))=y|Z=1-z,X=x\right)(1-p_{z|x})
\end{equation*}

which lies in the interval $\left[p_{y|z,x}p_{z|x},p_{y|z,x}p_{z|x}+(1-p_{z|x})\right]$, which concludes the proof.

\subsubsection*{Sharpness}

The proof is conducted the same way as in proposition 1. I begin by finding a DGP consistent with the data and Assumptions 1-3 that attains the upper bound.

Choose

\vspace{-10mm}

\begin{align*}
    &\mathbb{P}\left(Y(D(z))=y,D(z)=1|Z=1-z,X=x\right)=\frac{\mathbb{P}^{u}_{c}\left(Y(D(z))=y,D(z)=1|X=x\right)-p_{y,1|z,x}p_{z|x}}{1-p_{z|x}},\\&\mathbb{P}\left(Y(D(z))=y,D(z)=0|Z=1-z,X=x\right)=\frac{\mathbb{P}^{u}_{c}\left(Y(D(z))=y,D(z)=0|X=x\right)-p_{y,0|z,x}p_{z|x}}{1-p_{z|x}}
\end{align*}

Then we obtain

\vspace{-10mm}

\begin{align*}
    &\mathbb{P}\left(Y(D(z))=y|X=x\right)=\mathbb{P}\left(Y(D(z))=y,D(z)=1|X=x\right)+\mathbb{P}\left(Y(D(z))=y,D(z)=0|X=x\right)\\&=p_{y,1|z,x}p_{z|x}+\mathbb{P}\left(Y(D(z))=y,D(z)=1|Z=1-z,X=x\right)(1-p_{z|x})\\&+p_{y,0|z,x}p_{z|x}+\mathbb{P}\left(Y(D(z))=y,D(z)=0|Z=1-z,X=x\right)(1-p_{z|x})\\&=\mathbb{P}^{u}_{c}\left(Y(D(z))=y,D(z)=1|X=x\right)+\mathbb{P}^{u}_{c}\left(Y(D(z))=y,D(z)=0|X=x\right)\\&=\mathbb{P}^{u}_{c}\left(Y(D(z))=y|X=x\right)
\end{align*}

Also, note that under the additional restrictions, $\mathbb{P}^{u}_{c}\left(Y(D(z))=y,D(z)=d|X=x\right)=\frac{p_{y,d|z,x}p_{z|x}}{p_{z|x}-c}$, so it follows that

\vspace{-10mm}

\begin{equation*}
    \mathbb{P}\left(Y(D(z))=y|Z=1-z,X=x\right)=\frac{\frac{p_{y|z,x}p_{z|x}}{p_{z|x}-c}-p_{y|z,x}p_{z|x}}{1-p_{z|x}}\in\left(0,1\right)
\end{equation*}

Finally, $c$-dependence is satisfied because the bounds on the joint probabilities satisfy the inequality provided in \textit{Proof of 3} in Proposition 1. Therefore, there is a DGP consistent with the data and Assumptions 1-3 which attains the upper bound for potential outcomes.

A DGP consistent with the data and Assumptions 1-3 which attains the lower bound can be obtained analogously, and therefore, sharpness follows.

\subsection*{Proof of Proposition 3}

The proof is analogous to the one in Proposition 2.

\subsection*{Proof of Proposition 4}

\subsubsection*{Validity}

The ITT conditional on $X=x$ can be expressed as

\vspace{-10mm}

\begin{align*}
    &ITT_{x}=LATE_{co|x}\pi_{co|x}=\left\{\mathbb{P}\left(Y(1)=1|co,X=x\right)-\mathbb{P}\left(Y(0)=1|co,X=x\right)\right\}\pi_{co|x}\\&=\mathbb{P}\left(Y(1)=1,co|X=x\right)-\mathbb{P}\left(Y(0)=1,co|X=x\right)
\end{align*}

Using theorem 2 (i) from De Chaisemartin (2017), we write $ITT_{x}$ as

\vspace{-10mm}

\begin{equation*}
    ITT_{x}=\mathbb{P}\left(Y(D(1))=1|X=x\right)-\mathbb{P}\left(Y(D(0))=1|X=x\right)+LATE_{def|x}\pi_{def|x}
\end{equation*}

Therefore, it follows that

\vspace{-10mm}

\begin{equation*}
    ITT_{x}\leq \mathbb{P}^{u}\left(Y(D(1))=1|X=x\right)-\mathbb{P}^{l}\left(Y(D(0))=1|X=x\right)+\pi_{def|x}
\end{equation*}

and 

\vspace{-10mm}

\begin{equation*}
    ITT_{x}\geq \mathbb{P}^{l}\left(Y(D(1))=1|X=x\right)-\mathbb{P}^{u}\left(Y(D(0))=1|X=x\right)-\pi_{def|x}
\end{equation*}

\subsubsection*{Sharpness}

The proof is conducted with the same structure as the proofs of sharpness in the previous propositions. I begin by constructing a DGP that attains the upper bound.

Choose

\vspace{-10mm}

\begin{align*}
    &\mathbb{P}\left(Y(D(1))=1,D(1)=d|Z=0,X=x\right)=\frac{\mathbb{P}^{u}_{c}\left(Y(D(1))=1,D(1)=d|X=x\right)-p_{1|x}p_{1,d|1,x}}{1-p_{1|x}},\\&\mathbb{P}\left(Y(D(0))=1,D(0)=d|Z=1,X=x\right)=\frac{\mathbb{P}^{l}_{c}\left(Y(D(0))=1,D(0)=d|X=x\right)-p_{0|x}p_{1,d|0,x}}{1-p_{0|x}}
\end{align*}

It follows that from the choice, we have $\mathbb{P}\left(Y(D(1))=1|X=x\right)=\mathbb{P}^{u}_{c}\left(Y(D(1))=1|X=x\right)$ and $\mathbb{P}\left(Y(D(0))=1|X=x\right)=\mathbb{P}^{l}_{c}\left(Y(D(0))=1|X=x\right)$. Choose a defier share $\pi_{def|x}$ consistent with the Frechet bounds implies that the group shares

\vspace{-10mm}

\begin{align*}
    &\pi_{at|x}=\mathbb{P}\left(D(0)=1|X=x\right)-\pi_{def|x},\\&\pi_{nt|x}=\mathbb{P}\left(D(1)=0|X=x\right)-\pi_{def|x},\\&\pi_{c|x}=\mathbb{P}\left(D(1)=1|X=x\right)-\mathbb{P}\left(D(0)=1|X=x\right)+\pi_{def|x}
\end{align*}

\noindent are all nonnegative and sum up to one. Draw the compliance groups $g\in\left\{at,co,nt,def\right\}$ with probabilities $\pi_{g|x}$ described above.

Set potential treatments as usual:

\vspace{-10mm}

\begin{equation*}
    \left ( D(1),D(0) \right )=\left\{\begin{matrix}
\left ( 1,1 \right ), \text{for}\ $at$ 
\\\left ( 1,0 \right ), \text{for}\ $co$ 
 \\\left ( 0,0 \right ), \text{for}\ $nt$ 
 \\\left ( 0,1 \right ), \text{for}\ $def$ 
\end{matrix}\right.
\end{equation*}

We now set the potential outcomes. For the group of defiers, set $\left(Y(1),Y(0)\right)=(1,0)$ almost surely. For the remaining groups, set 

\vspace{-10mm}

{\small \begin{align*}
    &\mathbb{P}\left(Y(1)=1|at,X=x\right)=\frac{\mathbb{P}^{l}_{c}\left(Y(D(0))=1,D(0)=1|X=x\right)-\pi_{def|x}}{\pi_{at|x}},\\&\mathbb{P}\left(Y(0)=1|nt,X=x\right)=\frac{\mathbb{P}^{u}_{c}\left(Y(D(1))=1,D(1)=0|X=x\right)}{\pi_{nt|x}},\\&\mathbb{P}\left(Y(1)=1|co,X=x\right)=\frac{\mathbb{P}^{u}_{c}\left(Y(D(1))=1,D(1)=1|X=x\right)-\mathbb{P}^{l}_{c}\left(Y(D(0))=1,D(0)=1|X=x\right)+\pi_{def|x}}{\pi_{co|x}},\\&\mathbb{P}\left(Y(0)=1|co,x=X\right)=\frac{\mathbb{P}^{l}_{c}\left(Y(D(0))=1,D(0)=0|X=x\right)-\mathbb{P}^{u}_{c}\left(Y(D(1))=1,D(1)=0|X=x\right)}{\pi_{co|x}}
\end{align*}}

These are all probabilities that lie between 0 and 1. Under this construction, we obtain

\vspace{-10mm}

\begin{align*}
    &\mathbb{P}\left(Y(D(1))=1,D(1)=1|X=x\right)=\mathbb{P}^{u}_{c}\left(Y(D(1))=1,D(1)=1|X=x\right),\\&\mathbb{P}\left(Y(D(1))=1,D(1)=0|X=x\right)=\mathbb{P}^{u}_{c}\left(Y(D(1))=1,D(1)=0|X=x\right)
\end{align*}

And therefore, 

\vspace{-10mm}

\begin{align*}
    &\mathbb{P}\left(Y(D(1))=1|X=x\right)=\mathbb{P}^{u}_{c}\left(Y(D(1))=1,D(1)=1|X=x\right)+\mathbb{P}^{u}_{c}\left(Y(D(1))=1,D(1)=0|X=x\right)\\&=\mathbb{P}^{u}_{c}\left(Y(D(1))=1|X=x\right)
\end{align*}

Similarly, 

\vspace{-10mm}

\begin{align*}
    &\mathbb{P}\left(Y(D(0))=1,D(0)=1|X=x\right)=\mathbb{P}^{l}_{c}\left(Y(D(0))=1,D(0)=1|X=x\right),\\&\mathbb{P}\left(Y(D(0))=1,D(0)=0|X=x\right)=\mathbb{P}^{l}_{c}\left(Y(D(0))=1,D(0)=0|X=x\right)
\end{align*}

And therefore, 

\vspace{-10mm}

\begin{align*}
    &\mathbb{P}\left(Y(D(0))=1|X=x\right)=\mathbb{P}^{l}_{c}\left(Y(D(0))=1,D(0)=1|X=x\right)+\mathbb{P}^{l}_{c}\left(Y(D(0))=1,D(0)=0|X=x\right)\\&=\mathbb{P}^{l}_{c}\left(Y(D(0))=1|X=x\right)
\end{align*}

From which we conclude that

\vspace{-10mm}

\begin{align*}
    &ITT_{x}=\mathbb{P}\left(Y(D(1))=1|X=x\right)-\mathbb{P}\left(Y(D(0))=1|X=x\right)+LATE_{def|x}\pi_{def|x}\\&=\mathbb{P}^{u}_{c}\left(Y(D(1))=1|X=x\right)-\mathbb{P}^{l}_{c}\left(Y(D(0))=1|X=x\right)+\pi_{def|x}=ITT_{x}^{l}(c,\pi_{def|x})
\end{align*}

By construction, the joint probabilities are exactly the sharpness-attaining joint probabilities from Proposition 1, so the observed distribution and Assumptions 1-3 are respected. Assumption 4 holds because the chosen share of defiers is smaller than the implied share of compliers.

To construct a DGP which attains the lower bound, set

\vspace{-10mm}

\begin{align*}
    &\mathbb{P}\left(Y(D(1))=1,D(1)=d|Z=0,X=x\right)=\frac{\mathbb{P}^{l}_{c}\left(Y(D(1))=1,D(1)=d|X=x\right)-p_{1|x}p_{1,d|1,x}}{1-p_{1|x}},\\&\mathbb{P}\left(Y(D(0))=1,D(0)=d|Z=1,X=x\right)=\frac{\mathbb{P}^{u}_{c}\left(Y(D(0))=1,D(0)=d|X=x\right)-p_{0|x}p_{1,d|0,x}}{1-p_{0|x}}
\end{align*}

Keep the same group shares defined for the proof of the upper bound, and set $\left(Y(1),Y(0)\right)=\left(0,1\right)$ almost surely for defiers. For the remaining potential outcomes, set 

\vspace{-10mm}

{\small \begin{align*}
    &\mathbb{P}\left(Y(1)=1|at,X=x\right)=\frac{\mathbb{P}^{u}_{c}\left(Y(D(0))=1,D(0)=1|X=x\right)}{\pi_{at|x}},\\&\mathbb{P}\left(Y(0)=1|nt,X=x\right)=\frac{\mathbb{P}^{l}_{c}\left(Y(D(1))=1,D(1)=0|X=x\right)-\pi_{def|x}}{\pi_{nt|x}},\\&\mathbb{P}\left(Y(1)=1|co,X=x\right)=\frac{\mathbb{P}^{l}_{c}\left(Y(D(1))=1,D(1)=1|X=x\right)-\mathbb{P}^{u}_{c}\left(Y(D(0))=1,D(0)=1|X=x\right)}{\pi_{co|x}},\\&\mathbb{P}\left(Y(0)=1|co,x=X\right)=\frac{\mathbb{P}^{u}_{c}\left(Y(D(0))=1,D(0)=0|X=x\right)-\mathbb{P}^{l}_{c}\left(Y(D(1))=1,D(1)=0|X=x\right)+\pi_{def|x}}{\pi_{co|x}}
\end{align*}}

Under this construction, we obtain

\vspace{-10mm}

\begin{align*}
    &\mathbb{P}\left(Y(D(1))=1,D(1)=1|X=x\right)=\mathbb{P}^{l}_{c}\left(Y(D(1))=1,D(1)=1|X=x\right),\\&\mathbb{P}\left(Y(D(1))=1,D(1)=0|X=x\right)=\mathbb{P}^{l}_{c}\left(Y(D(1))=1,D(1)=0|X=x\right)
\end{align*}

And therefore, 

\vspace{-10mm}

\begin{align*}
    &\mathbb{P}\left(Y(D(1))=1|X=x\right)=\mathbb{P}^{l}_{c}\left(Y(D(1))=1,D(1)=1|X=x\right)+\mathbb{P}^{l}_{c}\left(Y(D(1))=1,D(1)=0|X=x\right)\\&=\mathbb{P}^{l}_{c}\left(Y(D(1))=1|X=x\right)
\end{align*}

Similarly, 

\vspace{-10mm}

\begin{align*}
    &\mathbb{P}\left(Y(D(0))=1,D(0)=1|X=x\right)=\mathbb{P}^{u}_{c}\left(Y(D(0))=1,D(0)=1|X=x\right),\\&\mathbb{P}\left(Y(D(0))=1,D(0)=0|X=x\right)=\mathbb{P}^{u}_{c}\left(Y(D(0))=1,D(0)=0|X=x\right)
\end{align*}

And therefore, 

\vspace{-10mm}

\begin{align*}
    &\mathbb{P}\left(Y(D(0))=1|X=x\right)=\mathbb{P}^{u}_{c}\left(Y(D(0))=1,D(0)=1|X=x\right)+\mathbb{P}^{u}_{c}\left(Y(D(0))=1,D(0)=0|X=x\right)\\&=\mathbb{P}^{u}_{c}\left(Y(D(0))=1|X=x\right)
\end{align*}

From which we conclude that

\vspace{-10mm}

\begin{align*}
    &ITT_{x}=\mathbb{P}\left(Y(D(1))=1|X=x\right)-\mathbb{P}\left(Y(D(0))=1|X=x\right)+LATE_{def|x}\pi_{def|x}\\&=\mathbb{P}^{l}_{c}\left(Y(D(1))=1|X=x\right)-\mathbb{P}^{u}_{c}\left(Y(D(0))=1|X=x\right)-\pi_{def|x}=ITT_{x}^{l}(c,\pi_{def|x})
\end{align*}

By construction, the joint probabilities are exactly the sharpness-attaining joint probabilities from Proposition 1, so the observed distribution and Assumptions 1-3 are respected. Assumption 4 holds because the chosen share of defiers is smaller than the implied share of compliers.

\subsection*{Proof of Proposition 5}

\subsubsection*{Validity}

Note that from Theorem 2(i) from De Chaisemartin (2017), we have that

\vspace{-10mm}

\begin{align*}
    &LATE_{co|x}=\frac{ITT_{x}}{\pi_{co|x}}=\frac{\mathbb{P}\left(Y(D(1))=1|X=x\right)-\mathbb{P}\left(Y(D(0))=1|X=x\right)+LATE_{def|x}\pi_{def|x}}{\mathbb{P}\left(D(1)=1|X=x\right)-\mathbb{P}\left(D(0)=1|X=x\right)+\pi_{def|x}}
\end{align*}

It follows that

Note that

\vspace{-10mm}

\begin{align*}
    &\pi_{co|x}\leq\min\left\{\mathbb{P}^{u}_{c}\left(D(1)=1|X=x\right)-\mathbb{P}^{l}_{c}\left(D(0)=1|X=x\right)+\pi_{def|x},1\right\},\\&\pi_{co|x}\geq\max\left\{\mathbb{P}^{l}_{c}\left(D(1)=1|X=x\right)-\mathbb{P}^{u}_{c}\left(D(0)=1|X=x\right)+\pi_{def|x},0\right\}
\end{align*}

Combining this inequalities with the bounds for $ITT_{x}$ from Proposition 4 yields the result.
\subsubsection*{Sharpness}

Note that if $c=0$, the joint potential probabilities $\mathbb{P}\left(Y(D(z))=y,D(z)=d|X=x\right)$ are point identified by $p_{y,d|z,x}$ for all $\left(y,d,z\right)\in\left\{0,1\right\}^{3}$. Therefore, it implies that 

\vspace{-10mm}

\begin{align*}
    &\mathbb{P}\left(Y=1|Z=1,X=x\right)-\mathbb{P}\left(Y=1|Z=0,X=x\right)-\pi_{def|x}\leq ITT_{x}\\&\leq \mathbb{P}\left(Y=1|Z=1,X=x\right)-\mathbb{P}\left(Y=1|Z=0,X=x\right)+\pi_{def|x},\\&\pi_{co|x}=\mathbb{P}\left(D=1|Z=1,X=x\right)-\mathbb{P}\left(D=1|Z=0,X=x\right)+\pi_{def|x}
\end{align*}

Define the compliance group shares by

\vspace{-10mm}

\begin{align*}
    &\pi_{at|x}=\mathbb{P}\left(D=1|Z=0,X=x\right)-\pi_{def|x},\\&\pi_{nt|x}=\mathbb{P}\left(D=0|Z=1,X=x\right)-\pi_{def|x},\\&\pi_{co|x}=\mathbb{P}\left(D=1|Z=1,X=x\right)-\mathbb{P}\left(D=1|Z=0,X=x\right)+\pi_{def|x}
\end{align*}

By the feasibility restriction on the share of the defiers, these shares are nonnegative and sum to one.

Draw the groups $g\in\left\{at,co,nt,def\right\}$ with probability $\pi_{g|x}$ independent of $Z$ conditional on $X=x$. Set potential treatments as

\vspace{-10mm}

\begin{equation*}
    \left ( D(1),D(0) \right )=\left\{\begin{matrix}
\left ( 1,1 \right ), \text{for}\ $at$ 
\\\left ( 1,0 \right ), \text{for}\ $co$ 
 \\\left ( 0,0 \right ), \text{for}\ $nt$ 
 \\\left ( 0,1 \right ), \text{for}\ $def$ 
\end{matrix}\right.
\end{equation*}

For defiers, set $\left(Y(1),Y(0)\right)=\left(1,0\right)$ almost surely. For the remaining potential outcomes for other compliance groups, set

\vspace{-10mm}

\begin{align*}
    &\mathbb{P}\left(Y(1)=1|at,X=x\right)=\frac{\mathbb{P}\left(Y=1,D=1|Z=0,X=x\right)-\pi_{def|x}}{\mathbb{P}\left(D=1|Z=0,X=x\right)-\pi_{def|x}},\\&\mathbb{P}\left(Y(0)=1|nt,X=x\right)=\frac{\mathbb{P}\left(Y=1,D=0|Z=1,X=x\right)}{\mathbb{P}\left(D=0|Z=1,X=x\right)-\pi_{def|x}},\\&\mathbb{P}\left(Y(1)=1|co,X=x\right)=\frac{\mathbb{P}\left(Y=1,D=1|Z=1,X=x\right)-\mathbb{P}\left(Y=1,D=1|Z=0,X=x\right)+\pi_{def|x}}{\mathbb{P}\left(D=1|Z=1,X=x\right)-\mathbb{P}\left(D=1|Z=0,X=x\right)+\pi_{def|x}},\\&\mathbb{P}\left(Y(0)=1|co,X=x\right)=\frac{\mathbb{P}\left(Y=1,D=0|Z=0,X=x\right)-\mathbb{P}\left(Y=1,D=0|Z=1,X=x\right)}{\mathbb{P}\left(D=1|Z=1,X=x\right)-\mathbb{P}\left(D=1|Z=0,X=x\right)+\pi_{def|x}}
\end{align*}

It is easy to see that the upper bound is attained and that the assumptions are satisfied.

To obtain the lower bound, keep the same choice for the compliance group shares and potential treatments. Set $\left(Y(1),Y(0)\right)=\left(0,1\right)$ for defier almost surely.

For the remaining potential outcomes for other compliance groups, set

\vspace{-10mm}

\begin{align*}
    &\mathbb{P}\left(Y(1)=1|at,X=x\right)=\frac{\mathbb{P}\left(Y=1,D=1|Z=0,X=x\right)}{\mathbb{P}\left(D=1|Z=0,X=x\right)-\pi_{def|x}},\\&\mathbb{P}\left(Y(0)=1|nt,X=x\right)=\frac{\mathbb{P}\left(Y=1,D=0|Z=1,X=x\right)-\pi_{def|x}}{\mathbb{P}\left(D=0|Z=1,X=x\right)-\pi_{def|x}},\\&\mathbb{P}\left(Y(1)=1|co,X=x\right)=\frac{\mathbb{P}\left(Y=1,D=1|Z=1,X=x\right)-\mathbb{P}\left(Y=1,D=1|Z=0,X=x\right)}{\mathbb{P}\left(D=1|Z=1,X=x\right)-\mathbb{P}\left(D=1|Z=0,X=x\right)+\pi_{def|x}},\\&\mathbb{P}\left(Y(0)=1|co,X=x\right)=\frac{\mathbb{P}\left(Y=1,D=0|Z=0,X=x\right)-\mathbb{P}\left(Y=1,D=0|Z=1,X=x\right)+\pi_{def|x}}{\mathbb{P}\left(D=1|Z=1,X=x\right)-\mathbb{P}\left(D=1|Z=0,X=x\right)+\pi_{def|x}}
\end{align*}

It is also to see that the lower bound is attained, and that the DGP is consistent with the data and assumptions. Hence, sharpness follows.

\subsection*{Proof of Proposition 6}

\subsubsection*{Validity}

Following Huber (2015), decompose $\mathbb{P}\left(Y(1)=1|X=x\right)$ as

\vspace{-10mm}

\begin{align*}
    &\mathbb{P}\left ( Y(1)=1|X=x \right )=\mathbb{P}\left ( Y(1)=1|X=x,at \right )\pi_{at|x}+\mathbb{P}\left ( Y(1)=1|X=x,co \right )\pi_{co|x}\\&+\mathbb{P}\left ( Y(1)=1|X=x,nt \right )\pi_{nt|x}+\mathbb{P}\left ( Y(1)=1|X=x,def \right )\pi_{def|x}\\&=\mathbb{P}(Y(D(1))=1,D(1)=1|X=x)+\mathbb{P}(Y(D(0))=1,D(0)=1|X=x)\\&+\left ( \mathbb{P}\left ( D(1)=0|X=x \right )-\pi_{def|x} \right )\mathbb{P}\left ( Y(1)=1|X=x, nt\right )\\&-\mathbb{P}\left ( Y(1)=1|X=x, at\right )\pi_{at|x}
\end{align*}

Combining this result with the worst-case upper bound for $\mathbb{P}\left(Y(1)=1|X=x\right)$ yields the proposed upper bound. Combining this result with the worst-case lower bound for $\mathbb{P}\left(Y(1)=1|X=x\right)$ yields the proposed lower bound.

Combining the \cite{huber2017jae} decomposition of $\mathbb{P}\left(Y(0)=1|X=x\right)$ with its worst-case bounds yields its identified set. Combining the bounds of potential outcomes yields the bounds for the ATE, which concludes the proof.

\subsubsection*{Sharpness}

I begin with the DGP that attains the upper bound.

Set $\pi_{at|x}=\mathbb{P}\left(D(0)=1|X=x\right)-\pi_{def|x}$, $\pi_{nt|x}=\mathbb{P}\left(D(1)=0|X=x\right)-\pi_{def|x}$ and $\pi_{co|x}=\mathbb{P}\left(D(1)=1|X=x\right)-\mathbb{P}\left(D(0)=1|X=x\right)+\pi_{def|x}$.

Set potential treatments as usual. Set the following joint probabilities:

\vspace{-10mm}

\begin{align*}
    &\mathbb{P}\left(Y(D(1))=1,D(1)=1|Z=0,X=x\right)=\frac{\mathbb{P}^{u}_{c}\left(Y(D(1))=1,D(1)=1|X=x\right)-p_{1,1|1,x}p_{1|x}}{1-p_{1|x}},\\&\mathbb{P}\left(Y(D(1))=1,D(1)=0|Z=0,X=x\right)=\frac{\mathbb{P}^{l}_{c}\left(Y(D(1))=1,D(1)=0|X=x\right)-p_{1,0|1,x}p_{1|x}}{1-p_{1|x}},\\&\mathbb{P}\left(Y(D(0))=1,D(0)=1|Z=1,X=x\right)=\frac{\mathbb{P}^{u}_{c}\left(Y(D(0))=1,D(0)=1|X=x\right)-p_{1,1|0,x}p_{0|x}}{1-p_{0|x}},\\&\mathbb{P}\left(Y(D(0))=1,D(0)=0|Z=1,X=x\right)=\frac{\mathbb{P}^{l}_{c}\left(Y(D(0))=1,D(0)=0|X=x\right)-p_{1,0|0,x}p_{0|x}}{1-p_{0|x}}
\end{align*}

Now, set potential outcomes. For the group of defiers, set $\left(Y(1),Y(0)\right)=\left(1,0\right)$ almost surely. For always-takers and never-takers, set $\mathbb{P}\left(Y(0)=1|at,X=x\right)=0$, and $\mathbb{P}\left(Y(1)=1|nt,X=x\right)=1$. For the remaining potential outcomes, set

\vspace{-10mm}

{\small \begin{align*}
    &\mathbb{P}\left(Y(1)=1|at,X=x\right)=\frac{\mathbb{P}^{u}_{c}\left(Y(D(0))=1,D(0)=1|X=x\right)-\pi_{def|x}}{\pi_{at|x}},\\&\mathbb{P}\left(Y(1)=1|co,X=x\right)=\frac{\mathbb{P}^{u}_{c}\left(Y(D(1))=1,D(1)=1|X=x\right)-\mathbb{P}^{u}_{c}\left(Y(D(0))=1,D(0)=1|X=x\right)+\pi_{def|x}}{\pi_{co|x}},\\&\mathbb{P}\left(Y(0)=1|nt,X=x\right)=\frac{\mathbb{P}^{l}_{c}\left(Y(D(1))=1,D(1)=0|X=x\right)}{\pi_{nt|x}},\\&\mathbb{P}\left(Y(0)=1|co,X=x\right)=\frac{\mathbb{P}^{l}_{c}\left(Y(D(0))=1,D(0)=0|X=x\right)-\mathbb{P}^{l}_{c}\left(Y(D(1))=1,D(1)=0|X=x\right)}{\pi_{co|x}}
\end{align*}}

Substituting these quantities in the expression of $ATE_{x}$ yields $ATE_{x}=ATE^{u}_{x}(c,\pi_{def|x})$. In order to construct the DGP that attains the lower bounds, set potential treatments and the compliance group shares the same way. Set the following joint probabilities:

\vspace{-10mm}

\begin{align*}
    &\mathbb{P}\left(Y(D(1))=1,D(1)=1|Z=0,X=x\right)=\frac{\mathbb{P}^{l}_{c}\left(Y(D(1))=1,D(1)=1|X=x\right)-p_{1,1|1,x}p_{1|x}}{1-p_{1|x}},\\&\mathbb{P}\left(Y(D(1))=1,D(1)=0|Z=0,X=x\right)=\frac{\mathbb{P}^{u}_{c}\left(Y(D(1))=1,D(1)=0|X=x\right)-p_{1,0|1,x}p_{1|x}}{1-p_{1|x}},\\&\mathbb{P}\left(Y(D(0))=1,D(0)=1|Z=1,X=x\right)=\frac{\mathbb{P}^{l}_{c}\left(Y(D(0))=1,D(0)=1|X=x\right)-p_{1,1|0,x}p_{0|x}}{1-p_{0|x}},\\&\mathbb{P}\left(Y(D(0))=1,D(0)=0|Z=1,X=x\right)=\frac{\mathbb{P}^{u}_{c}\left(Y(D(0))=1,D(0)=0|X=x\right)-p_{1,0|0,x}p_{0|x}}{1-p_{0|x}}
\end{align*}

Now, set potential outcomes. For the group of defiers, set $\left(Y(1),Y(0)\right)=\left(0,1\right)$ almost surely. For always-takers and never-takers, set $\mathbb{P}\left(Y(0)=1|at,X=x\right)=1$, and $\mathbb{P}\left(Y(1)=1|nt,X=x\right)=0$. For the remaining potential outcomes, set

\vspace{-10mm}

{\small \begin{align*}
    &\mathbb{P}\left(Y(1)=1|at,X=x\right)=\frac{\mathbb{P}^{l}_{c}\left(Y(D(0))=1,D(0)=1|X=x\right)}{\pi_{at|x}},\\&\mathbb{P}\left(Y(1)=1|co,X=x\right)=\frac{\mathbb{P}^{l}_{c}\left(Y(D(1))=1,D(1)=1|X=x\right)-\mathbb{P}^{l}_{c}\left(Y(D(0))=1,D(0)=1|X=x\right)}{\pi_{co|x}},\\&\mathbb{P}\left(Y(0)=1|nt,X=x\right)=\frac{\mathbb{P}^{u}_{c}\left(Y(D(1))=1,D(1)=0|X=x\right)-\pi_{def|x}}{\pi_{nt|x}},\\&\mathbb{P}\left(Y(0)=1|co,X=x\right)=\frac{\mathbb{P}^{u}_{c}\left(Y(D(0))=1,D(0)=0|X=x\right)-\mathbb{P}^{u}_{c}\left(Y(D(1))=1,D(1)=0|X=x\right)+\pi_{def|x}}{\pi_{co|x}}
\end{align*}}

Substituting these quantities in the expression of $ATE_{x}$ yields $ATE_{x}=ATE^{l}_{x}(c,\pi_{def|x})$. Therefore, sharpness follows.

\subsection*{Proof of Corollary 1}

I show that the bounds for the ATE coincide with the bounds from  \cite{Balke01091997} and \cite{Chen1015-7483R1} when the sensitivity parameters are set to 0. I begin with the upper bound. Note that if $c=0$, then

\vspace{-10mm}

\begin{align*}
    &ATE^{u}_{x}=\mathbb{P}\left(Y=1,D=1|Z=1,X=x\right)+\mathbb{P}\left(Y=1,D=1|Z=0,X=x\right)\\&-\mathbb{P}\left(Y=1,D=0|Z=0,X=x\right)-\mathbb{P}\left(Y=1,D=0|Z=1,X=x\right)\\&+\min\left\{\mathbb{P}\left(Y=1,D=0|Z=0,X=x\right),\mathbb{P}\left(Y=1,D=0|Z=0,X=x\right)\right\}\\&+\left(\mathbb{P}\left(D=0|Z=1,X=x\right)-\pi_{d|x}\right)
\end{align*}

If we further assume that $\pi_{d|x}=0$, then it follows from \cite{toru} that

\vspace{-10mm}

\begin{equation*}
    \min\left\{\mathbb{P}\left(Y=1,D=0|Z=0,X=x\right),\mathbb{P}\left(Y=1,D=1|Z=0,X=x\right)\right\}=\mathbb{P}\left(Y=1,D=0|Z=1,X=x\right)
\end{equation*}

and because the upper bound of $\mathbb{P}\left(Y(1)=1|X=x\right)$ is achieved by setting $\mathbb{P}\left(Y(1)=1|at,X=x\right)=0$, it follows that 

\vspace{-10mm}

\begin{equation*}
    \mathbb{P}\left(Y=1,D=1|Z=0,X=x\right)=0
\end{equation*}

And thus, the upper bound becomes

\vspace{-10mm}

\begin{equation*}
    ATE^{u}_{x}=\mathbb{P}\left(Y=1,D=1|Z=1,X=x\right)-\mathbb{P}\left(Y=1,D=0|Z=0,X=x\right)+\mathbb{P}\left(D=0|Z=1,X=x\right)
\end{equation*}

In the case of the lower bound we apply the same reasoning, but this time we use the fact that

\vspace{-10mm}

\begin{align*}
    & \min\left\{\mathbb{P}\left(Y=1,D=1|Z=1,X=x\right),\mathbb{P}\left(Y=1,D=1|Z=0,X=x\right)\right\}=\mathbb{P}\left(Y=1,D=1|Z=0,X=x\right),\\&\mathbb{P}\left(Y=1,D=0|Z=1,X=x\right)=0
\end{align*}

where the first result follows from \cite{toru} and the second from the fact that the lower bound of $\mathbb{P}\left(Y(0)=1|X=x\right)$ is achieved by setting $\mathbb{P}\left(Y(0)=1|nt,X=x\right)=0$. And therefore, 

\vspace{-10mm}

\begin{equation*}
    ATE^{u}_{x}=\mathbb{P}\left(Y=1,D=1|Z=1,X=x\right)-\mathbb{P}\left(Y=1,D=0|Z=0,X=x\right)-\mathbb{P}\left(D=1|Z=0,X=x\right)
\end{equation*}

\noindent which concludes the proof.

\subsection*{Proof of Theorem 1}

Recall that $\widehat{bf}(c,\mu)=\widehat{\mathbb{P}}^{l}\left(Y(D(1))=1\right)-\widehat{\mathbb{P}}^{u}\left(Y(D(0))=1\right)-\mu$. By lemma 3, we know that $\widehat{\mathbb{P}}^{l}\left(Y(D(1))=1\right)-\widehat{\mathbb{P}}^{u}\left(Y(D(0))=1\right)$ converges uniformly over $c\in\mathcal{C}$. Lemma 3 further implies that 

\vspace{-10mm}

\begin{align*}
    &\sqrt{N}\left(\widehat{\mathbb{P}}^{l}\left(Y(D(1))=1\right)-\widehat{\mathbb{P}}^{u}\left(Y(D(0))=1\right)-\left(\mathbb{P}^{l}\left(Y(D(1))=1\right)-\mathbb{P}^{u}\left(Y(D(0))=1\right)\right)\right)\\&\xrightarrow[d]{}\sum_{k=1}^{K}q_{x_{k}}\left(\textbf{Z}_{Y}^{(2)}(1,1,x_{k},c)-\textbf{Z}_{Y}^{(1)}(1,0,x_{k},c)\right)+\sum_{k=1}^{K}\left(\mathbb{P}^{l}\left(Y(D(1))=1\right)-\mathbb{P}^{u}\left(Y(D(0))=1\right)\right)\textbf{Z}_{1}^{(3)}(0,0,0,x_{k})\\&\equiv\tilde{\textbf{Z}}(c)
\end{align*}

\noindent where $\tilde{\textbf{Z}}(c)$ is  a random element of $l^{\infty}\left(\mathcal{C}\right)$. And thus, $\sqrt{N}\left(\widehat{bf}(c,\mu)-bf(c,\mu)\right)$ converges to a random element in $l^{\infty}\left(\mathcal{C}\times\mathcal{M}\right)$. Therefore, by the delta method for
Hadamard directionally differentiable functions, $\sqrt{N}\left(\widehat{BF}(c,\mu)-BF(c,\mu)\right)$ converges in process, which concludes the proof. 

\subsection*{Proof of Theorem 2}

The proof is analogous to Theorem 1.

\section*{Appendix B}

\begin{lemma}
Suppose Assumptions 5 and 6 hold. Then, 

\vspace{-5mm}

\begin{equation*}
    \sqrt{N}\begin{pmatrix}
\widehat{p}_{y,d|z,x}-p_{y,d|z,x}
 \\\widehat{p}_{z|x}-p_{z|x}
 \\\widehat{q}_{x}-q_{x}
\end{pmatrix} \xrightarrow[d]{}\textbf{Z}_{1}(y,d,z,x)
\end{equation*}

\noindent a mean-zero Gaussian process in $l^{\infty}\left(\left\{0,1\right\}^{3}\times\mathcal{S}(X),\mathbb{R}^{3}\right)$. with covariance kernel $\Sigma_{1}$ defined in the proof.
    
\end{lemma}

\textbf{Proof:}

By a second-order Taylor Expansion, we obtain

\vspace{-10mm}

\begin{align*}
    &\widehat{p}_{y,d|z,x}-p_{y,d|z,x}=\frac{\frac{1}{N}\sum_{i=1}^{N}\mathbf{1}\left\{Y_{i}=y,D_{i}=d\right\}\mathbf{1}\left\{Z_{i}=z,X_{i}=x\right\}}{\frac{1}{N}\sum_{i=1}^{N}\mathbf{1}\left\{Z_{i}=z,X_{i}=x\right\}}-\frac{\mathbb{P}\left(Y=y,D=d,Z=z,X=x\right)}{\mathbb{P}\left(Z=z,X=x\right)}\\&=\frac{\frac{1}{N}\sum_{i=1}^{N}\mathbf{1}\left\{Y_{i}=y,D_{i}=d\right\}\mathbf{1}\left\{Z_{i}=z,X_{i}=x\right\}-\mathbb{P}\left(Y=y,D=d,Z=z,X=x\right)}{\mathbb{P}\left(Z=z,X=x\right)}\\&-\frac{\mathbb{P}\left(Y=y,D=d|Z=z,X=x\right)}{\mathbb{P}\left(Z=z,X=x\right)}\left(\frac{1}{N}\sum_{i=1}^{N}\mathbf{1}\left\{Z_{i}=z,X_{i}=x\right\}-\mathbb{P}\left(Z_{i}=z,X_{i}=x\right)\right)+O_{p}\left(N^{-1}\right)\\&=\frac{1}{N}\sum_{i=1}^{N}\frac{\mathbf{1}\left\{Z_{i}=z,X_{i}=x\right\}\left(\mathbf{1}\left\{Y_{i}=y,D_{i}=d\right\}-p_{y,d|z,x}\right)}{\mathbb{P}\left(Z=z,X=x\right)}+o_{p}\left(N^{-1/2}\right)
\end{align*}

\noindent and hence, $\sqrt{N}\left(\widehat{p}_{y,d|z,x}-p_{y,d|z,x}\right)$ converges in distribution to a mean-zero Gaussian process with continuous paths. Similarly, one obtains the following linear representations:

\vspace{-10mm}

\begin{align*}
   &\widehat{p}_{z|x}-p_{z|x}=\frac{1}{N}\sum_{i=1}^{N}\frac{\mathbf{1}\left\{X_{i}=x\right\}\left(\mathbf{1}\left\{Z_{i}=z\right\}-p_{z|x}\right)}{q_{x}}+o_{p}\left(N^{-1/2}\right),\\&\widehat{q}_{x}-q_{x}=\frac{1}{N}\sum_{i=1}^{N}\left(\mathbf{1}\left\{X_{i}=x\right\}-q_{x}\right)
\end{align*}

The covariance kernel $\Sigma_{1}$ has diagonal elements respectively equal to

\vspace{-10mm}

\begin{align*}
    &\frac{p_{y,d|z,x}(1-p_{y,d|z,x})}{p_{z|x}q_{x}},\\&\mathbb{E}\left [ \frac{\mathbf{1}\left\{X_{i}=x \right\}\mathbf{1}\left\{ X_{i}=\tilde{x}\right\}\left ( \mathbf{1}\left\{ Z_{i}=z\right\}-p_{z|x} \right )\left ( \mathbf{1}\left\{ Z_{i}=\tilde{z}\right\}-p_{\tilde{z}|\tilde{x}} \right )}{q_{x}q_{\tilde{x}}} \right ],\\&q_{x}q_{\tilde{x}}
\end{align*}

and all remaining elements equal to zero, which completes the proof.

\begin{lemma}
    Suppose Assumptions 1-7 hold. Then,
\begin{equation*}
     \sqrt{N}\begin{pmatrix}
\widehat{\mathbb{P}}^{u}\left(Y(D(z))=y,D(z)=d|X=x\right)-\mathbb{P}^{u}\left(Y(D(z))=y,D(z)=d|X=x\right) 
\\\widehat{\mathbb{P}}^{l}\left(Y(D(z))=y,D(z)=d|X=x\right)-\mathbb{P}^{l}\left(Y(D(z))=y,D(z)=d|X=x\right)
\end{pmatrix}\xrightarrow[d]{}\textbf{Z}_{Y,D}(y,d,z,x,c)
 \end{equation*}

\noindent  a tight element of $l^{\infty}\left ( \left\{ 0,1\right\}^{3}\times\mathcal{S}(X),\mathbb{R}^{2} \right )$.
\end{lemma}

\textbf{Proof:}

Let $\theta_{0}=\left ( p_{y,d|z,x},p_{d|z,x},p_{z|x},q_{x} \right )$ and $\widehat{\theta}=\left ( \widehat{p}_{y,d|z,x},\widehat{p}_{d|z,x},\widehat{p}_{z|x},\widehat{q}_{x} \right )$. For fixed $y$, $d$ and fixed $c$, define the mapping

\vspace{-10mm}

\begin{equation*}
    \phi:l^{\infty}\left ( \left\{ 0,1\right\}^{3}\times \mathcal{S}(X) \right )\times l^{\infty}\left ( \left\{ 0,1\right\}\times \mathcal{S}(X) \right )\times l^{\infty}\left ( \mathcal{S}(X) \right )\rightarrow l^{\infty}\left ( \left\{ 0,1\right\}^{2},\mathcal{S}(X),\mathbb{R}^{2} \right )
\end{equation*}

by

\vspace{-10mm}

\begin{equation*}
    \left [ \phi(\theta) \right ](z,x)=\begin{pmatrix}
\min\left\{ \frac{\theta^{(1)}(y,d,z,x)\theta^{(3)}(z,x)}{\theta^{(3)}(z,x)-c},\frac{\theta^{(1)}(y,d,z,x)\theta^{(3)}(z,x)+c}{\theta^{(3)}(z,x)+c},\theta^{(1)}(y,d,z,x)\theta^{(3)}(z,x)+(1-\theta^{(3)}(z,x))\right\} \\
\max\left\{ \frac{\theta^{(1)}(y,d,z,x)\theta^{(3)}(z,x)}{\theta^{(3)}(z,x)+c},\frac{\theta^{(1)}(y,d,z,x)\theta^{(3)}(z,x)-c}{\theta^{(3)}(z,x)-c},\theta^{(1)}(y,d,z,x)\theta^{(3)}(z,x)\right\}
\end{pmatrix}
\end{equation*}

where $\theta^{(j)}$ is the j-th component of $\theta$. Note that

\vspace{-10mm}

\begin{equation*}
    \left [ \phi(\theta_{0}) \right ](z,x)=\begin{pmatrix}
\mathbb{P}^{u}\left(Y(D(z))=y,D(z)=d|X=x\right) \\
\mathbb{P}^{l}\left(Y(D(z))=y,D(z)=d|X=x\right)
\end{pmatrix}
\end{equation*}

The mapping $\phi$ is comprised with max and min operators, along with six other functions. We begin by computing the Hadamard derivative of these functions with respect to $\theta$ using \cite{fangsantos} and the Chain rule for Hadamard differentiable functions to obtain the derivative of $\phi$.

Let $h\in\mathbb{R}^{2}$. First, consider $\left [ \delta_{1}(\theta) \right ](z,x)=\frac{\theta^{(1)}(y,d,z,x)\theta^{(3)}(z,x)}{\theta^{(3)}(z,x)+c}$, which has Hadamard derivative equal to

\vspace{-10mm}

\begin{equation*}
    \left [ \delta_{1,\theta}^{'}(h) \right ](z,x)=\frac{\theta^{(1)}(y,d,z,x)h^{(3)}(z,x)+h^{(1)}(y,d,z,x)\theta^{(3)}(z,x)}{\theta^{(3)}(z,x)+c}-\frac{\theta^{(1)}(y,d,z,x)\theta^{(3)}(z,x)h^{(3)}(z,x)}{\left ( \theta^{(3)}(z,x)+c \right )^{2}}
\end{equation*}

Next, $\left [ \delta_{2}(\theta) \right ](z,x)=\frac{\theta^{(1)}(y,d,z,x)\theta^{(3)}(z,x)-c}{\theta^{(3)}(z,x)-c}$ has Hadamard derivative equal to

\vspace{-10mm}

\begin{equation*}
    \left [ \delta_{2,\theta}^{'}(h) \right ](z,x)=\frac{\theta^{(1)}(y,d,z,x)h^{(3)}(z,x)+h^{(1)}(y,d,z,x)\theta^{(3)}(z,x)}{\theta^{(3)}(z,x)-c}-\frac{\left ( \theta^{(1)}(y,d,z,x)\theta^{(3)}(z,x)-c_{} \right )h^{(3)}(z,x)}{\left ( \theta^{(3)}(z,x)-c_{} \right )^{2}}
\end{equation*}

Next, $\left [ \delta_{3}(\theta) \right ](z,x)=\theta^{(1)}(y,d,z,x)\theta^{(3)}(z,x)$ has Hadamard derivative equal to

\vspace{-10mm}

\begin{equation*}
    \left [ \delta_{3,\theta}^{'}(h) \right ](z,x)=h^{(1)}(y,d,z,x)\theta^{(3)}(z,x)+\theta^{(1)}(y,d,z,x)h^{(3)}(z,x)
\end{equation*}

Now, we turn to the functionals inside the $\min$ operator. First, we have $\left [ \delta_{4}(\theta) \right ](z,x)=\frac{\theta^{(1)}(y,d,z,x)\theta^{(3)}(z,x)}{\theta^{(3)}(z,x)-c_{}}$, which has Hadamard derivative equal to

\vspace{-10mm}

\begin{equation*}
    \left [ \delta_{4,\theta}^{'}(h) \right ](z,x)=\frac{\theta^{(1)}(y,d,z,x)h^{(3)}(z,x)+h^{(1)}(y,d,z,x)\theta^{(3)}(z,x)}{\theta^{(3)}(z,x)-c_{}}-\frac{\theta^{(1)}(y,d,z,x)\theta^{(3)}(z,x)h^{(3)}(z,x)}{\left ( \theta^{(3)}(z,x)-c_{} \right )^{2}}
\end{equation*}

Next, $\left [ \delta_{5}(\theta) \right ](z,x)=\frac{\theta^{(1)}(y,d,z,x)\theta^{(3)}(z,x)+c_{}}{\theta^{(3)}(z,x)+c_{}}$ has Hadamard derivative equal to

\vspace{-10mm}

\begin{equation*}
    \left [ \delta_{5,\theta}^{'}(h) \right ](z,x)=\frac{\theta^{(1)}(y,d,z,x)h^{(3)}(z,x)+h^{(1)}(y,d,z,x)\theta^{(3)}(z,x)}{\theta^{(3)}(z,x)+c_{}}-\frac{\left ( \theta^{(1)}(y,d,z,x)\theta^{(3)}(z,x)+c_{} \right )h^{(3)}(z,x)}{\left ( \theta^{(3)}(z,x)+c_{} \right )^{2}}
\end{equation*}

Finally, $\left [ \delta_{6,\theta}^{'}(h) \right ](z,x)=h^{(1)}(y,d,z,x)\theta^{(3)}(z,x)+h^{(3)}(z,x)(\theta^{(1)}(y,d,z,x)-1)$.

Using this notation, we write the functional $\phi_{}$ as

\vspace{-10mm}

\begin{equation*}
    \phi_{}(\theta)=\begin{pmatrix}
\min\left\{ \delta_{4}(\theta),\delta_{5}(\theta),\delta_{6}(\theta\right\} \\
\max\left\{ \delta_{1}(\theta),\delta_{2}(\theta),\delta_{3}(\theta)\right\}
\end{pmatrix}
\end{equation*}

Using the chain rule \citep{mastenpoirier2020}, the Hadamard derivative of $\phi_{}$ at $\theta_{0}$ is

\vspace{-10mm}

\begin{equation*}
    \phi_{\theta_{0}}^{'}(h)=\begin{pmatrix}
\mathbf{1}\left ( \delta_{6}(\theta_{0})>\min\left\{\delta_{4}(\theta_{0}),\delta_{5}(\theta_{0}) \right\} \right )\delta_{6,\theta_{0}}^{'}(h)\\
+\mathbf{1}\left ( \delta_{5}(\theta_{0})>\min\left\{\delta_{4}(\theta_{0}),\delta_{6}(\theta_{0}) \right\} \right )\delta_{5,\theta_{0}}^{'}(h) \\
+\mathbf{1}\left ( \delta_{4}(\theta_{0})>\max\left\{\delta_{5}(\theta_{0}),\delta_{6}(\theta_{0}) \right\} \right )\delta_{4,\theta_{0}}^{'}(h)\\
+\mathbf{1}\left ( \delta_{6}(\theta_{0})=\delta_{5}(\theta_{0})>\delta_{4}(\theta_{0})  \right )\min\left\{ \delta_{6,\theta_{0}}^{'}(h),\delta_{5,\theta_{0}}^{'}(h)\right\}\\
+\mathbf{1}\left ( \delta_{6}(\theta_{0})=\delta_{4}(\theta_{0})>\delta_{5}(\theta_{0})  \right )\min\left\{ \delta_{6,\theta_{0}}^{'}(h),\delta_{4,\theta_{0}}^{'}(h)\right\}\\
+\mathbf{1}\left ( \delta_{4}(\theta_{0})=\delta_{5}(\theta_{0})>\delta_{6}(\theta_{0})  \right )\min\left\{ \delta_{4,\theta_{0}}^{'}(h),\delta_{5,\theta_{0}}^{'}(h)\right\}\\
+\mathbf{1}\left ( \delta_{6}(\theta_{0})=\delta_{5}(\theta_{0})=\delta_{4}(\theta_{0})  \right )\min\left\{ \delta_{6,\theta_{0}}^{'}(h),\delta_{5,\theta_{0}}^{'}(h),\delta_{4,\theta_{0}}^{'}(h)\right\}\\
 \\
\mathbf{1}\left ( \delta_{1}(\theta_{0})<\max\left\{\delta_{1}(\theta_{0}),\delta_{2}(\theta_{0}) \right\} \right )\delta_{1,\theta_{0}}^{'}(h)\\
+\mathbf{1}\left ( \delta_{2}(\theta_{0})<\max\left\{\delta_{1}(\theta_{0}),\delta_{3}(\theta_{0}) \right\} \right )\delta_{2,\theta_{0}}^{'}(h) \\
+\mathbf{1}\left ( \delta_{3}(\theta_{0})<\max\left\{\delta_{1}(\theta_{0}),\delta_{2}(\theta_{0}) \right\} \right )\delta_{3,\theta_{0}}^{'}(h)\\
+\mathbf{1}\left ( \delta_{3}(\theta_{0})=\delta_{2}(\theta_{0})<\delta_{1}(\theta_{0})  \right )\max\left\{ \delta_{3,\theta_{0}}^{'}(h),\delta_{2,\theta_{0}}^{'}(h)\right\}\\
+\mathbf{1}\left ( \delta_{3}(\theta_{0})=\delta_{1}(\theta_{0})<\delta_{2}(\theta_{0})  \right )\max\left\{ \delta_{3,\theta_{0}}^{'}(h),\delta_{1,\theta_{0}}^{'}(h)\right\}\\
+\mathbf{1}\left ( \delta_{1}(\theta_{0})=\delta_{2}(\theta_{0})<\delta_{3}(\theta_{0})  \right )\max\left\{ \delta_{1,\theta_{0}}^{'}(h),\delta_{2,\theta_{0}}^{'}(h)\right\}\\
+\mathbf{1}\left ( \delta_{3}(\theta_{0})=\delta_{2}(\theta_{0})=\delta_{1}(\theta_{0})  \right )\max\left\{ \delta_{3,\theta_{0}}^{'}(h),\delta_{2,\theta_{0}}^{'}(h),\delta_{1,\theta_{0}}^{'}(h)\right\}\\
\end{pmatrix}
\end{equation*}

By Lemma 1, $\sqrt{N}\left ( \widehat{\theta}-\theta_{0} \right )\xrightarrow[d]{}\textbf{Z}_{1}(y,d,z,x)$. Using the Delta Method for Hadamard differentiable functions, we obtain

\vspace{-10mm}

\begin{equation*}
    \left [ \sqrt{N}\left ( \phi(\widehat{\theta})-\phi(\theta_{0}) \right ) \right ](z,x)\xrightarrow[d]{}\left [ \phi_{\theta_{0}}^{'}(\textbf{Z}_{1}) \right ](z,x)\equiv \tilde{\textbf{Z}}_{Y,D}(z,x)
\end{equation*}

This result holds uniformly over any finite grid of values for $d\in\left\{0,1\right\}^{}$ and $c\in\mathcal{C}$ by considering the Hadamard directional differentiability of a vector of these mappings indexed at
different values of $y$ and $c$, which yields the process $\textbf{Z}_{Y,D}(y,d,z,x,c)$.

\begin{lemma}
    Suppose Assumptions 1-7 hold. Then,

    \vspace{-10mm}

 \begin{equation*}
     \sqrt{N}\begin{pmatrix}
\widehat{\mathbb{P}}^{u}\left(Y(D(z))=y|X=x\right)-\mathbb{P}^{u}\left(Y(D(z))=y|X=x\right) 
\\\widehat{\mathbb{P}}^{l}\left(Y(D(z))=y|X=x\right)-\mathbb{P}^{l}\left(Y(D(z))=y|X=x\right)
\end{pmatrix}\xrightarrow[d]{}\textbf{Z}_{Y}(y,z,x,c)
 \end{equation*}

\noindent  a tight element of $l^{\infty}\left ( \left\{ 0,1\right\}^{2}\times\mathcal{S}(X)\times\mathcal{C},\mathbb{R}^{2} \right )$.
\end{lemma}

\textbf{Proof:}

Recall that

\vspace{-10mm}

{\small \begin{align*}
    &\widehat{\mathbb{P}}_{c}^{u}\left ( Y(D(z))=y|X=x \right )\\&=\min\left\{\widehat{\mathbb{P}}^{u}_{c}\left(Y(D(z))=y,D(z)=1|X=x\right)+\widehat{\mathbb{P}}^{u}_{c}\left(Y(D(z))=y,D(z)=0|X=x\right),\widehat{p}_{y|z,x}\widehat{p}_{z|x}+(1-\widehat{p}_{z|x})\right\}\\&\equiv\min\left\{\widehat{\psi}_{Y}^{(1)}(y,z,x,c),\widehat{\psi}_{Y}^{(2)}(y,z,x)\right\}
\end{align*}}

From Lemma 2, we have that

\vspace{-10mm}

\begin{align*}
    &\sqrt{N}\left(\widehat{\psi}_{Y}^{(1)}(y,z,x,c)-\psi_{Y}^{(1)}(y,z,x,c)\right)\xrightarrow[d]{}\textbf{Z}_{Y,D}^{(1)}(y,1,z,x,c)+\textbf{Z}_{Y,D}^{(1)}(y,0,z,x,c)\\&\equiv\textbf{Z}_{Y}^{(1,1)}(y,z,x,c)
\end{align*}

From Lemma 1, we have 

\vspace{-10mm}

\begin{align*}
    &\sqrt{N}\left(\widehat{\psi}_{Y}^{(2)}(y,z,x)-\psi_{Y}^{(2)}(y,z,x)\right)\xrightarrow[d]{}p_{z|x}\left(\textbf{Z}_{1}^{(1)}(y,1,z,x)+\textbf{Z}_{1}^{(1)}(y,0,z,x)\right)+\left(p_{y|z,x}-1\right)\textbf{Z}_{1}^{(2)}(z,x)\\&\equiv\textbf{Z}_{Y}^{(1,2)}(y,z,x)
\end{align*}

Applying Theorem 2.1 from \cite{fangsantos}, we find that 

\vspace{-10mm}

\begin{align*}
    &\sqrt{N}\left(\widehat{\mathbb{P}}^{u}_{c}\left(Y(D(z))=y|X=x\right)-\mathbb{P}^{u}_{c}\left(Y(D(z))=y|X=x\right)\right)\\&\xrightarrow[d]{}\left\{\begin{matrix}
\textbf{Z}_{Y}^{(1,1)}(y,z,x,c)\mathbf{1}\left\{ \psi_{Y}^{(1)}(y,z,x,c)<\psi_{Y}^{(2)}(y,z,x)\right\}
 \\\min\left\{ \textbf{Z}_{Y}^{(1,1)}(y,z,x,c),\textbf{Z}_{Y}^{(1,2)}(y,z,x)\right\}\mathbf{1}\left\{ \psi_{Y}^{(1)}(y,z,x,c)=\psi_{Y}^{(2)}(y,z,x)\right\}
 \\\textbf{Z}_{Y}^{(1,2)}(y,z,x,c)\mathbf{1}\left\{ \psi_{Y}^{(1)}(y,z,x,c)>\varphi_{Y}^{(2)}(y,z,x)\right\}
\end{matrix}\right.\equiv\textbf{Z}_{Y}^{(1)}(y,z,x,c)
\end{align*}

For the lower bound, we have

\vspace{-10mm}

{\small \begin{align*}
    &\widehat{\mathbb{P}}_{c}^{l}\left ( Y(D(z))=y|X=x \right )\\&=\max\left\{\widehat{\mathbb{P}}^{l}_{c}\left(Y(D(z))=y,D(z)=1|X=x\right)+\widehat{\mathbb{P}}^{l}_{c}\left(Y(D(z))=y,D(z)=0|X=x\right),\widehat{p}_{y|z,x}\widehat{p}_{z|x}\right\}\\&\equiv\max\left\{\widehat{\varphi}_{Y}^{(1)}(y,z,x,c),\widehat{\varphi}_{Y}^{(2)}(y,z,x)\right\}
\end{align*}}

From Lemma 2, we have that

\vspace{-10mm}

\begin{align*}
    &\sqrt{N}\left(\widehat{\varphi}_{Y}^{(1)}(y,z,x,c)-\varphi_{Y}^{(1)}(y,z,x,c)\right)\xrightarrow[d]{}\textbf{Z}_{Y,D}^{(2)}(y,1,z,x,c)+\textbf{Z}_{Y,D}^{(2)}(y,0,z,x,c)\\&\equiv\textbf{Z}_{Y}^{(2,1)}(y,z,x,c)
\end{align*}

From Lemma 1, we have 

\vspace{-10mm}

\begin{align*}
    &\sqrt{N}\left(\widehat{\varphi}_{Y}^{(2)}(y,z,x)-\varphi_{Y}^{(2)}(y,z,x)\right)\xrightarrow[d]{}p_{z|x}\left(\textbf{Z}_{1}^{(1)}(y,1,z,x)+\textbf{Z}_{1}^{(1)}(y,0,z,x)\right)+p_{y|z,x}\textbf{Z}_{1}^{(2)}(z,x)\\&\equiv\textbf{Z}_{Y}^{(2,2)}(y,z,x)
\end{align*}

Applying Theorem 2.1 from \cite{fangsantos}, we find that 

\vspace{-10mm}

\begin{align*}
    &\sqrt{N}\left(\widehat{\mathbb{P}}^{l}_{c}\left(Y(D(z))=y|X=x\right)-\mathbb{P}^{l}_{c}\left(Y(D(z))=y|X=x\right)\right)\\&\xrightarrow[d]{}\left\{\begin{matrix}
\textbf{Z}_{Y}^{(2,1)}(y,z,x,c)\mathbf{1}\left\{ \varphi_{Y}^{(1)}(y,z,x,c)>\varphi_{Y}^{(2)}(y,z,x)\right\}
 \\\max\left\{ \textbf{Z}_{Y}^{(2,1)}(y,z,x,c),\textbf{Z}_{Y}^{(2,2)}(y,z,x)\right\}\mathbf{1}\left\{ \varphi_{Y}^{(1)}(y,z,x,c)=\varphi_{Y}^{(2)}(y,z,x)\right\}
 \\\textbf{Z}_{Y}^{(2,2)}(y,z,x,c)\mathbf{1}\left\{ \varphi_{Y}^{(1)}(y,z,x,c)<\varphi_{Y}^{(2)}(y,z,x)\right\}
\end{matrix}\right.\equiv\textbf{Z}_{Y}^{(2)}(y,z,x,c)
\end{align*}

And therefore, we have

\vspace{-10mm}

\begin{align*}
    &\sqrt{N}\begin{pmatrix}
\widehat{\mathbb{P}}^{u}_{c}\left ( Y(D(z))=y|X=x \right )-\mathbb{P}^{u}_{c}\left ( Y(D(z))=y|X=x \right )
 \\\widehat{\mathbb{P}}^{l}_{c}\left ( Y(D(z))=y|X=x \right )-\mathbb{P}^{l}_{c}\left ( Y(D(z))=y|X=x \right )
\end{pmatrix}\xrightarrow[d]{}\begin{pmatrix}
\textbf{Z}_{Y}^{(1)}(y,z,x,c)
 \\\textbf{Z}_{Y}^{(2)}(y,z,x,c)
\end{pmatrix}\\&\equiv \textbf{Z}_{Y}(y,z,x,c)
\end{align*}

which concludes the proof.

\begin{lemma}
    Suppose Assumptions 1-7 hold. Then,

    \vspace{-5mm}

 \begin{equation*}
     \sqrt{N}\begin{pmatrix}
\widehat{\mathbb{P}}^{u}\left(D(z)=d|X=x\right)-\mathbb{P}^{u}\left(D(z)=d|X=x\right) 
\\\widehat{\mathbb{P}}^{l}\left(D(z)=d|X=x\right)-\mathbb{P}^{l}\left(D(z)=d|X=x\right)
\end{pmatrix}\xrightarrow[d]{}\textbf{Z}_{D}(d,z,x,c)
 \end{equation*}

\noindent  a tight element of $l^{\infty}\left ( \left\{ 0,1\right\}^{2}\times\mathcal{S}(X)\times\mathcal{C},\mathbb{R}^{2} \right )$.
\end{lemma}

\textbf{Proof:}

\vspace{-10mm}

{\small \begin{align*}
    &\widehat{\mathbb{P}}_{c}^{u}\left ( D(z)=d|X=x \right )\\&=\min\left\{\widehat{\mathbb{P}}^{u}_{c}\left(Y(D(z))=1,D(z)=d|X=x\right)+\widehat{\mathbb{P}}^{u}_{c}\left(Y(D(z))=0,D(z)=d|X=x\right),\widehat{p}_{d|z,x}\widehat{p}_{z|x}+(1-\widehat{p}_{z|x})\right\}\\&\equiv\min\left\{\widehat{\psi}_{D}^{(1)}(d,z,x,c),\widehat{\psi}_{D}^{(2)}(d,z,x)\right\}
\end{align*}}

From Lemma 2, we have that

\vspace{-10mm}

\begin{align*}
    &\sqrt{N}\left(\widehat{\psi}_{D}^{(1)}(d,z,x,c)-\psi_{D}^{(1)}(d,z,x,c)\right)\xrightarrow[d]{}\textbf{Z}_{Y,D}^{(1)}(1,d,z,x,c)+\textbf{Z}_{Y,D}^{(1)}(0,d,z,x,c)\\&\equiv\textbf{Z}_{D}^{(1,1)}(d,z,x,c)
\end{align*}

From Lemma 1, we have 

\vspace{-10mm}

\begin{align*}
    &\sqrt{N}\left(\widehat{\psi}_{D}^{(2)}(d,z,x)-\psi_{D}^{(2)}(d,z,x)\right)\xrightarrow[d]{}p_{z|x}\left(\textbf{Z}_{1}^{(1)}(1,d,z,x)+\textbf{Z}_{1}^{(1)}(0,d,z,x)\right)+\left(p_{d|z,x}-1\right)\textbf{Z}_{1}^{(2)}(z,x)\\&\equiv\textbf{Z}_{D}^{(1,2)}(d,z,x)
\end{align*}

Applying Theorem 2.1 from \cite{fangsantos}, we find that 

\vspace{-10mm}

\begin{align*}
    &\sqrt{N}\left(\widehat{\mathbb{P}}^{u}_{c}\left(D(z)=d|X=x\right)-\mathbb{P}^{u}_{c}\left(D(z)=d|X=x\right)\right)\\&\xrightarrow[d]{}\left\{\begin{matrix}
\textbf{Z}_{D}^{(1,1)}(d,z,x,c)\mathbf{1}\left\{ \psi_{D}^{(1)}(d,z,x,c)<\psi_{D}^{(2)}(d,z,x)\right\}
 \\\min\left\{ \textbf{Z}_{D}^{(1,1)}(d,z,x,c),\textbf{Z}_{D}^{(1,2)}(d,z,x)\right\}\mathbf{1}\left\{ \psi_{Y}^{(1)}(d,z,x,c)=\psi_{D}^{(2)}(d,z,x)\right\}
 \\\textbf{Z}_{D}^{(1,2)}(d,z,x,c)\mathbf{1}\left\{ \psi_{D}^{(1)}(d,z,x,c)>\varphi_{D}^{(2)}(d,z,x)\right\}
\end{matrix}\right.\equiv\textbf{Z}_{D}^{(1)}(d,z,x,c)
\end{align*}

For the lower bound, we have

\vspace{-10mm}

{\small \begin{align*}
    &\widehat{\mathbb{P}}_{c}^{l}\left ( D(z)=d|X=x \right )\\&=\max\left\{\widehat{\mathbb{P}}^{l}_{c}\left(Y(D(z))=1,D(z)=d|X=x\right)+\widehat{\mathbb{P}}^{l}_{c}\left(Y(D(z))=0,D(z)=d|X=x\right),\widehat{p}_{d|z,x}\widehat{p}_{z|x}\right\}\\&\equiv\max\left\{\widehat{\varphi}_{D}^{(1)}(d,z,x,c),\widehat{\varphi}_{D}^{(2)}(d,z,x)\right\}
\end{align*}}

From Lemma 2, we have that

\vspace{-10mm}

\begin{align*}
    &\sqrt{N}\left(\widehat{\varphi}_{D}^{(1)}(d,z,x,c)-\varphi_{D}^{(1)}(d,z,x,c)\right)\xrightarrow[d]{}\textbf{Z}_{Y,D}^{(2)}(1,d,z,x,c)+\textbf{Z}_{Y,D}^{(2)}(0,d,z,x,c)\\&\equiv\textbf{Z}_{D}^{(2,1)}(d,z,x,c)
\end{align*}

From Lemma 1, we have 

\vspace{-10mm}

\begin{align*}
    &\sqrt{N}\left(\widehat{\varphi}_{D}^{(2)}(d,z,x)-\varphi_{D}^{(2)}(d,z,x)\right)\xrightarrow[d]{}p_{z|x}\left(\textbf{Z}_{1}^{(1)}(1,d,z,x)+\textbf{Z}_{1}^{(1)}(0,d,z,x)\right)+p_{d|z,x}\textbf{Z}_{1}^{(2)}(z,x)\\&\equiv\textbf{Z}_{D}^{(2,2)}(d,z,x)
\end{align*}

Applying Theorem 2.1 from \cite{fangsantos}, we find that 

\vspace{-10mm}

\begin{align*}
    &\sqrt{N}\left(\widehat{\mathbb{P}}^{l}_{c}\left(D(z)=d|X=x\right)-\mathbb{P}^{l}_{c}\left(D(z)=d|X=x\right)\right)\\&\xrightarrow[d]{}\left\{\begin{matrix}
\textbf{Z}_{D}^{(2,1)}(d,z,x,c)\mathbf{1}\left\{ \varphi_{Y}^{(1)}(d,z,x,c)>\varphi_{Y}^{(2)}(d,z,x)\right\}
 \\\max\left\{ \textbf{Z}_{D}^{(2,1)}(d,z,x,c),\textbf{Z}_{D}^{(2,2)}(d,z,x)\right\}\mathbf{1}\left\{ \varphi_{D}^{(1)}(d,z,x,c)=\varphi_{D}^{(2)}(d,z,x)\right\}
 \\\textbf{Z}_{D}^{(2,2)}(d,z,x,c)\mathbf{1}\left\{ \varphi_{D}^{(1)}(d,z,x,c)<\varphi_{D}^{(2)}(d,z,x)\right\}
\end{matrix}\right.\equiv\textbf{Z}_{D}^{(2)}(d,z,x,c)
\end{align*}

And therefore, we have

\vspace{-10mm}

\begin{align*}
    &\sqrt{N}\begin{pmatrix}
\widehat{\mathbb{P}}^{u}_{c}\left ( D(z)=d|X=x \right )-\mathbb{P}^{u}_{c}\left ( D(z)=d|X=x \right )
 \\\widehat{\mathbb{P}}^{l}_{c}\left ( D(z)=d|X=x \right )-\mathbb{P}^{l}_{c}\left ( D(z)=d|X=x \right )
\end{pmatrix}\xrightarrow[d]{}\begin{pmatrix}
\textbf{Z}_{D}^{(1)}(d,z,x,c)
 \\\textbf{Z}_{D}^{(2)}(d,z,x,c)
\end{pmatrix}\\&\equiv \textbf{Z}_{D}(d,z,x,c)
\end{align*}

which concludes the proof.

\begin{lemma}
    Suppose Assumptions 1-7 hold. Then, 

\vspace{-10mm}

\begin{equation*}
    \sqrt{N}\begin{pmatrix}
\widehat{ITT}^{u}(c,\pi_{def})-ITT_{}^{u}(c,\pi_{def})
\\\widehat{ITT}_{}^{l}(c,\pi_{def})-ITT_{}^{l}(c,\pi_{def})
\end{pmatrix}
\xrightarrow[d]{}\Breve{\textbf{Z}}_{ITT}(y,z,x,c,\pi_{def})
\end{equation*}

\noindent a tight element of $l^{\infty}\left ( \left\{ 0,1\right\}^{2}\times\mathcal{S}(X)\times\mathcal{C},\mathbb{R}^{2} \right )$.

\end{lemma}

\textbf{Proof:} From Lemma 3, it follows that

\vspace{-10mm}

\begin{equation*}
    \sqrt{N}\begin{pmatrix}
\widehat{\mathbb{P}}^{u}\left ( Y(D(1))=1|X=x \right )-\mathbb{P}^{u}\left ( Y(D(1))=1|X=x \right )
 \\\widehat{\mathbb{P}}^{l}\left ( Y(D(1))=1|X=x \right )-\mathbb{P}^{l}\left ( Y(D(1))=1|X=x \right )
 \\\widehat{\mathbb{P}}^{u}\left ( Y(D(0))=1|X=x \right )-\mathbb{P}^{u}\left ( Y(D(0))=1|X=x \right )
 \\\widehat{\mathbb{P}}^{l}\left ( Y(D(0))=1|X=x \right )-\mathbb{P}^{l}\left ( Y(D(0))=1|X=x \right )
\end{pmatrix}\xrightarrow[d]{}\begin{pmatrix}
 \textbf{Z}_{Y}^{(1)}(1,1,x,c)
\\ \textbf{Z}_{Y}^{(2)}(1,1,x,c)
\\ \textbf{Z}_{Y}^{(1)}(1,0,x,c)
\\ \textbf{Z}_{Y}^{(2)}(1,0,x,c)
\end{pmatrix}\equiv \ddot{\textbf{Z}}_{Y}(y,z,x,c)
\end{equation*}

Let $\widehat{\gamma}$ denote the estimates for the bounds of potential outcomes and $\gamma_{0}$ their population values. For fixed $y$, $c$ and $\pi_{def|x}$, define the mapping 

\vspace{-10mm}

\begin{equation*}
    \phi_{ITT_{x}}:l^{\infty}\left ( \left\{ 0,1\right\}^{2}\times \mathcal{S}(X) \right )\times l^{\infty}\left ( \left\{ 0,1\right\}\times \mathcal{S}(X) \right )\times l^{\infty}\left ( \mathcal{S}(X) \right )\rightarrow l^{\infty}\left ( \left\{ 0,1\right\},\mathcal{S}(X),\mathbb{R}^{2} \right )
\end{equation*}

by

\vspace{-10mm}

\begin{equation*}
    \left [ \phi_{ITT_{x}}(\theta) \right ](z,x)=\begin{pmatrix}
\min\left\{\gamma^{(1)}(y,z,x,c)-\gamma^{(4)}(y,z,x,c)+\pi_{d|x},1 \right\}
\\\max\left\{\gamma^{(2)}(y,z,x,c)-\gamma^{(3)}(y,z,x,c)-\pi_{d|x},-1 \right\}
\end{pmatrix}
\end{equation*}

The Hadamard derivative for $\left [ \delta_{1}(\gamma)  \right ](z,x)=\gamma^{(1)}(y,z,x,c)-\gamma^{(4)}(y,z,x,c)+\pi_{d|x}$ is

\vspace{-10mm}

\begin{equation*}
    \left [ \delta^{'}_{1,\gamma}(h)  \right ](z,x)=h^{(1)}(y,z,x,c)-h^{(4)}(y,z,x,c)
\end{equation*}

The Hadamard derivative for $\left [ \delta_{2}(\gamma)  \right ](z,x)=1$ is equal to 0. The Hadamard derivative for $\left [ \delta_{3}(\gamma)  \right ](z,x)=\gamma^{(2)}(y,z,x,c)-\gamma^{(3)}(y,z,x,c)+\pi_{def|x}$ is

\vspace{-10mm}

\begin{equation*}
    \left [ \delta^{'}_{1,\gamma}(h)  \right ](z,x)=h^{(2)}(y,z,x,c)-h^{(3)}(y,z,x,c)
\end{equation*}

\noindent and finally, the The Hadamard derivative for $\left [ \delta_{4}(\gamma)  \right ](z,x)=-1$ is equal to 0. Hence, the Hadamard directional derivative of $\phi_{ITT}$ evaluated at $\gamma_{0}$ is

\vspace{-10mm}

\begin{equation*}
    \phi_{ITT_{x},\gamma_{0}}^{'}(h)=\begin{pmatrix}
\mathbf{1}\left ( \delta_{1}(\gamma_{0})=1 \right )\min\left\{ \delta^{'}_{1,\gamma_{0}}(h),0\right\} \\
+\mathbf{1}\left ( \delta_{1}(\gamma_{0})<1 \right )\delta^{'}_{1,\gamma_{0}}(h) \\ 
 \\\mathbf{1}\left ( \delta_{3}(\gamma_{0})=-1 \right )\max\left\{ \delta^{'}_{3,\gamma_{0}}(h),0\right\}
 \\+\mathbf{1}\left ( \delta_{,3}(\gamma_{0})>-1 \right )\delta^{'}_{3,\gamma_{0}}(h) 
\end{pmatrix}
\end{equation*}

By Lemma 3 and the Delta Method for Hadamard directionally differentiable functions,

\vspace{-10mm}

\begin{equation*}
    \sqrt{N}\left ( \phi_{ITT_{x}}(\widehat{\gamma})-\phi_{ITT_{x}}(\gamma_{0}) \right )\overset{d}{\rightarrow}\left [ \phi_{\gamma_{0}}^{'}(\ddot{\textbf{Z}}) \right ](z,x)\equiv\tilde{\textbf{Z}}_{ITT}(z,x)
\end{equation*}

which yields the process $\textbf{Z}_{ITT}(y,z,x,c,\pi_{d|x})$. It follows directly that the estimator for the unconditional upper bound converges weakly to a Gaussian element:

\vspace{-10mm}

{\small \begin{align*}
    &\sqrt{N}\left(\widehat{ITT}^{u}(c,\pi)-ITT^{u}(c,\pi)\right)=\sqrt{N}\sum_{k=1}^{K}\widehat{q}_{x_{k}}\left(\widehat{ITT}^{u}_{x_{k}}(c,\pi_{d|x_{k}})-ITT^{u}_{x{k}}(c,\pi_{d|x_{k}})\right)+\sum_{k=1}^{K}ITT_{x_{k}}\sqrt{N}\left(\widehat{q}_{x_{k}}-q_{x_{k}}\right)\\&\xrightarrow[d]{}\sum_{k=1}^{K}q_{x_{k}}\textbf{Z}_{ITT}^{(1)}(y,z,x_{k},c,\pi_{d|x_{k}})+\sum_{k=1}^{K}ITT^{u}_{x_{k}}(c,\pi_{d|x_{k}})\textbf{Z}_{1}^{(3)}(0,0,0,x_{k})\equiv \Breve{\textbf{Z}}^{(1)}_{ITT}
\end{align*}}

A similar result holds yields

\begin{equation*}
    \sqrt{N}\left(\widehat{ITT}^{l}-ITT^{l}\right)\xrightarrow[d]{} \Breve{\textbf{Z}}^{(2)}_{ITT}
\end{equation*}

which concludes the proof.

\begin{lemma}
    Suppose Assumptions 1-7 hold. Then, 

    \vspace{-10mm}

\begin{equation*}
    \sqrt{N}\begin{pmatrix}
\widehat{\pi}_{co}^{u}(c,\pi_{def})-\pi^{u}_{co}(c,\pi_{def})
\\\widehat{\pi}_{co}^{l}(c,\pi_{def})-\pi_{co}^{l}(c,\pi_{def})
\end{pmatrix}
\xrightarrow[d]{}\Breve{\textbf{Z}}_{FS}(d,z,x,c,\pi_{def})
\end{equation*}

\noindent a tight element of $l^{\infty}\left ( \left\{ 0,1\right\}^{2}\times\mathcal{S}(X)\times\mathcal{C},\mathbb{R}^{2} \right )$.
\end{lemma}

\textbf{Proof:} From Lemma 4, it follows that

\vspace{-10mm}

\begin{equation*}
    \sqrt{N}\begin{pmatrix}
\widehat{\mathbb{P}}^{u}\left ( D(1)=1|X=x \right )-\mathbb{P}^{u}\left ( D(1)=1|X=x \right )
 \\\widehat{\mathbb{P}}^{l}\left ( D(1)=1|X=x \right )-\mathbb{P}^{l}\left ( D(1)=1|X=x \right )
 \\\widehat{\mathbb{P}}^{u}\left ( D(0)=1|X=x \right )-\mathbb{P}^{u}\left ( D(0)=1|X=x \right )
 \\\widehat{\mathbb{P}}^{l}\left ( D(0)=1|X=x \right )-\mathbb{P}^{l}\left ( D(0)=1|X=x \right )
\end{pmatrix}\xrightarrow[d]{}\begin{pmatrix}
 \textbf{Z}_{D}^{(1)}(1,1,x,c)
\\ \textbf{Z}_{D}^{(2)}(1,1,x,c)
\\ \textbf{Z}_{D}^{(1)}(1,0,x,c)
\\ \textbf{Z}_{D}^{(2)}(1,0,x,c)
\end{pmatrix}\equiv \ddot{\textbf{Z}}_{D}(d,z,x,c)
\end{equation*}

Let $\widehat{\gamma}$ denote the estimates for the bounds of potential outcomes and $\gamma_{0}$ their population values. For fixed $d$, $c$ and $\pi_{def|x}$, define the mapping 

\vspace{-10mm}

\begin{equation*}
    \phi_{FS_{x}}:l^{\infty}\left ( \left\{ 0,1\right\}^{2}\times \mathcal{S}(X) \right )\times l^{\infty}\left ( \left\{ 0,1\right\}\times \mathcal{S}(X) \right )\times l^{\infty}\left ( \mathcal{S}(X) \right )\rightarrow l^{\infty}\left ( \left\{ 0,1\right\},\mathcal{S}(X),\mathbb{R}^{2} \right )
\end{equation*}

by

\vspace{-10mm}

\begin{equation*}
    \left [ \phi_{FS_{x}}(\theta) \right ](z,x)=\begin{pmatrix}
\min\left\{\gamma^{(1)}(d,z,x,c)-\gamma^{(4)}(d,z,x,c)+\pi_{d|x},1 \right\}
\\\max\left\{\gamma^{(2)}(d,z,x,c)-\gamma^{(3)}(d,z,x,c)+\pi_{d|x},0 \right\}
\end{pmatrix}
\end{equation*}

The Hadamard derivative for $\left [ \delta_{1}(\gamma)  \right ](z,x)=\gamma^{(1)}(d,z,x,c)-\gamma^{(4)}(d,z,x,c)+\pi_{def|x}$ is

\begin{equation*}
    \left [ \delta^{'}_{1,\gamma}(h)  \right ](z,x)=h^{(1)}(d,z,x,c)-h^{(4)}(d,z,x,c)
\end{equation*}

The Hadamard derivative for $\left [ \delta_{2}(\gamma)  \right ](z,x)=1$ is equal to 0. The Hadamard derivative for $\left [ \delta_{3}(\gamma)  \right ](z,x)=\gamma^{(2)}(d,z,x,c)-\gamma^{(3)}(d,z,x,c)+\pi_{def|x}$ is

\vspace{-10mm}

\begin{equation*}
    \left [ \delta^{'}_{3,\gamma}(h)  \right ](z,x)=h^{(2)}d,z,x,c)-h^{(3)}(d,z,x,c)
\end{equation*}

\noindent and finally, the The Hadamard derivative for $\left [ \delta_{4}(\gamma)  \right ](z,x)=0$ is equal to 0. Hence, the Hadamard directional derivative of $\phi_{FS_{x}}$ evaluated at $\gamma_{0}$ is

\vspace{-10mm}

\begin{equation*}
    \phi_{FS_{x},\gamma_{0}}^{'}(h)=\begin{pmatrix}
\mathbf{1}\left ( \delta_{1}(\gamma_{0})=1 \right )\min\left\{ \delta^{'}_{1,\gamma_{0}}(h),0\right\} \\
+\mathbf{1}\left ( \delta_{1}(\gamma_{0})<1 \right )\delta^{'}_{1,\gamma_{0}}(h) \\ 
 \\\mathbf{1}\left ( \delta_{3}(\gamma_{0})=-1 \right )\max\left\{ \delta^{'}_{3,\gamma_{0}}(h),0\right\}
 \\+\mathbf{1}\left ( \delta_{,3}(\gamma_{0})>-1 \right )\delta^{'}_{3,\gamma_{0}}(h) 
\end{pmatrix}
\end{equation*}

By Lemma 4 and the Delta Method for Hadamard directionally differentiable functions,

\vspace{-10mm}

\begin{equation*}
    \sqrt{N}\left ( \phi_{FS_{x}}(\widehat{\gamma})-\phi_{FS_{x}}(\gamma_{0}) \right )\overset{d}{\rightarrow}\left [ \phi_{\gamma_{0}}^{'}(\ddot{\textbf{Z}}) \right ](z,x)\equiv\tilde{\textbf{Z}}_{FS}(z,x)
\end{equation*}

which yields the process $\textbf{Z}_{FS}(d,z,x,c,\pi_{def|x})$. It follows directly that the estimator for the unconditional upper bound converges weakly to a Gaussian element:

\vspace{-10mm}

{\small\begin{align*}
    &\sqrt{N}\left(\widehat{\pi}_{co}^{u}(c,\pi_{def})-\pi_{co}^{u}(c,\pi_{def})\right)=\sqrt{N}\sum_{k=1}^{K}\widehat{q}_{x_{k}}\left(\widehat{\pi}^{u}_{co|x_{k}}(c,\pi_{def|x_{k})})-\pi^{u}_{co|x{k}}(c,\pi_{def|x_{k}})\right)+\sum_{k=1}^{K}\pi_{co|x_{k}}\sqrt{N}\left(\widehat{q}_{x_{k}}-q_{x_{k}}\right)\\&\xrightarrow[d]{}\sum_{k=1}^{K}q_{x_{k}}\textbf{Z}_{FS}^{(1)}(d,z,x_{k},c,\pi_{def|x_{k}})+\sum_{k=1}^{K}\pi^{u}_{co|x_{k}}(c,\pi_{def|x_{k}})\textbf{Z}_{1}^{(3)}(0,0,0,x_{k})\equiv\Breve{\textbf{Z}}^{(1)}_{FS}
\end{align*}}

A similar result holds for the estimator of the unconditional lower bound for the first-stage, which concludes the proof.

\begin{lemma}
    Suppose Assumptions 1-7 hold. Then, 

    \vspace{-10mm}

    \begin{equation*}
        \sqrt{N}\begin{pmatrix}
\widehat{LATE}^{u}(c,\pi_{def})-LATE^{u}(c,\pi_{def})
 \\\widehat{LATE}^{l}(c,\pi_{def})-LATE^{l}(c,\pi_{def})
\end{pmatrix}\xrightarrow[d]{}\textbf{Z}_{LATE}(y,d,z,x,c,\pi_{def})
    \end{equation*}

\noindent a tight element of $l^{\infty}\left ( \left\{ 0,1\right\}^{2}\times\mathcal{S}(X)\times\mathcal{C},\mathbb{R}^{2} \right )$.
\end{lemma}

\textbf{Proof:}

From Lemmas 5 and 6, we have

\vspace{-10mm}

\begin{equation*}
    \sqrt{N}\begin{pmatrix}
\widehat{ITT}^{u}(c,\pi_{def})-ITT^{u}(c,\pi_{def})
 \\\widehat{ITT}^{l}(c,\pi_{def})-ITT^{l}(c,\pi_{def})
 \\\widehat{FS}^{u}(c,\pi_{def})-FS^{u}(c,\pi_{def})
 \\\widehat{FS}^{l}(c,\pi_{def})-FS^{l}(c,\pi_{def})
\end{pmatrix}\xrightarrow[d]{}\begin{pmatrix}
\breve{\textbf{Z}}^{(1)}_{ITT}(y,z,x,c,\pi_{def})
\\\breve{\textbf{Z}}^{(2)}_{ITT}(y,z,x,c,\pi_{def})
\\\breve{\textbf{Z}}^{(1)}_{FS}(d,z,x,c,\pi_{def})
\\\breve{\textbf{Z}}^{(2)}_{FS}(d,z,x,c,\pi_{def})
\end{pmatrix}\equiv \breve{\textbf{Z}}(y,d,z,x,c,\pi_{def})
\end{equation*}

Let $\widehat{\theta}$ denote the estimated parameters above and $\theta_{0}$ denote its population values. For fixed $y,d,c,\pi_{def}$, define the mapping

\vspace{-10mm}

\begin{align*}
    &\phi_{LATE}:l^{\infty}\left ( \left\{ 0,1\right\}^{3}\times\mathcal{S}(X) \right )\times l^{\infty}\left ( \left\{ 0,1\right\}^{2}\times\mathcal{S}(X) \right )\times l^{\infty}\left ( \left\{ 0,1\right\}\times\mathcal{S}(X) \right )\times l^{\infty}\left ( \mathcal{S}(X) \right )\\&\rightarrow l^{\infty}\left ( \left\{ -1,1\right\},\mathcal{S}(X),\mathbb{R}^{2} \right )
\end{align*}

by

\vspace{-10mm}

\begin{equation*}
    \left [ \phi_{LATE}(\theta_{}) \right ](z,x)=\begin{pmatrix}
\min\left\{ \frac{\theta_{}^{(1)}(y,z,x,c_{},\pi_{def})}{\theta_{}^{(4)}(d,z,x,c_{},\pi_{def})},1\right\} \\
\max\left\{ \frac{\theta_{}^{(2)}(y,z,x,c_{},\pi_{def})}{\theta_{}^{(3)}(d,z,x,c_{},\pi_{def})},-1\right\}
\end{pmatrix}
\end{equation*}

The Hadamard derivative for $\left [ \delta_{1}(\theta) \right ](z,x)=\frac{\theta^{(1)}(y,z,x,c_{},\pi_{def})}{\theta^{(4)}(d,z,x,c_{},\pi_{def})}$ is equal to

\vspace{-10mm}

\begin{equation*}
    \left [ \delta^{'}_{1,\theta}(h) \right ](z,x)=\frac{h^{(1)}(y,z,x,c_{},\pi_{def})\theta^{(4)}(d,z,x,c_{},\pi_{def})-\theta^{(1)}(y,z,x,c_{},\pi_{def})h^{(4)}(d,z,x,c_{},\pi_{def})}{\left ( \theta^{(4)}(d,z,x,c_{},\pi_{def}) \right )^{2}}
\end{equation*}

The Hadamard derivative for $\left [ \delta_{2}(\theta) \right ](z,x)=1$ is equal to $0$. The hadamard derivative for $\left [ \delta_{3}(\theta) \right ](z,x)=\frac{\theta^{(2)}(y,z,x,c_{},\pi_{def})}{\theta^{(3)}(d,z,x,c_{},\pi_{def})}$ is equal to

\vspace{-10mm}

\begin{equation*}
    \left [ \delta^{'}_{3,\theta}(h) \right ](z,x)=\frac{h^{(2)}(y,z,x,c_{},\pi_{def})\theta^{(3)}(d,z,x,c_{},\pi_{def})-\theta^{(2)}(y,z,x,c_{},\pi_{def})h^{(3)}(d,z,x,c_{},\pi_{def})}{\left ( \theta^{(3)}(d,z,x,c_{},\pi_{def}) \right )^{2}}
\end{equation*}

And the Hadamard derivative for $\left [ \delta_{4}(\theta) \right ](z,x)=-1$ is equal to $0$.

Hence, the Hadamard directional derivative of $\phi_{LATE}$ evaluated at $\theta_{0}$ is 

\vspace{-10mm}

\begin{equation*}
    \phi_{LATE,\theta_{0}}^{'}(h)=\begin{pmatrix}
\mathbf{1}\left ( \delta_{1}(\theta_{0})=1 \right )\min\left\{ \delta^{'}_{1,\theta_{0}}(h),0\right\} \\
+\mathbf{1}\left ( \delta_{1}(\theta_{0})<1 \right )\delta^{'}_{1,\theta_{0}}(h) \\ 
 \\\mathbf{1}\left ( \delta_{3}(\theta_{0})=-1 \right )\max\left\{ \delta^{'}_{3,\theta_{0}}(h),0\right\}
 \\+\mathbf{1}\left ( \delta_{3}(\theta_{0})>-1 \right )\delta^{'}_{3,\theta_{0}}(h) 
\end{pmatrix}
\end{equation*}

By the Delta Method for Hadamard directionally differentiable functions,

\vspace{-10mm}

\begin{equation*}
    \sqrt{N}\left ( \phi_{LATE}(\widehat{\theta})-\phi_{LATE}(\theta_{0}) \right )\overset{d}{\rightarrow}\left [ \phi_{LATE,\theta_{0}}^{'}(\Breve{\textbf{Z}}) \right ](z,x)\equiv\tilde{\textbf{Z}}_{LATE}(z,x)
\end{equation*}

which yields the process $\textbf{Z}_{LATE}(y,d,z,x,c_{},\pi)$.

\section*{Appendix C}

In this section I describe the DGP from Section 4.2.1, which is also the DGP used to conduct the Monte Carlo simulations. 

First, I consider a single binary covariate, which follows a Bernoulli distribution with parameter $p=0.5$. Therefore, $\mathbb{P}\left(X=1\right)=\mathbb{P}\left(X=0\right)=0.5$.

The conditional distribution of the instrument is also a Bernoulli. I set $p_{1|x}=0.6$ for $x\in\left\{0,1\right\}$.

The joint distribution of outcome, treatment and assignment is the same for both values of the covariate. Hence, I omit the covariate for the sake of the exposition.

I set

\vspace{-10mm}

\begin{align*}
    &\mathbb{P}\left(Y=1,D=1,Z=1\right)=0.2,\\&\mathbb{P}\left(Y=1,D=0,Z=1\right)=0.1,\\&\mathbb{P}\left(Y=0,D=1,Z=1\right)=0.25,\\&\mathbb{P}\left(Y=0,D=0,Z=1\right)=0.05
\end{align*}

and 

\vspace{-10mm}

\begin{align*}
    &\mathbb{P}\left(Y=1,D=1,Z=0\right)=0.05,\\&\mathbb{P}\left(Y=1,D=0,Z=0\right)=0.05,\\&\mathbb{P}\left(Y=0,D=1,Z=0\right)=0.05,\\&\mathbb{P}\left(Y=0,D=0,Z=0\right)=0.25
\end{align*}

which generates the values displayed in Section 4.2.1

\end{document}